\newif\ifec
\newif\ifarxiv
\newif\ifcolor
\newif\ifhide
\ecfalse
\arxivfalse
\colorfalse
\hidetrue

\ifec
\documentclass[acmsmall]{ec20acm}

\copyrightyear{2020} 
\acmYear{2020} 
\setcopyright{acmcopyright}\acmConference[EC '20]{Proceedings of the 21st ACM Conference on Economics and Computation}{July 13--17, 2020}{Virtual Event, Hungary}
\acmBooktitle{Proceedings of the 21st ACM Conference on Economics and Computation (EC '20), July 13--17, 2020, Virtual Event, Hungary}
\acmPrice{15.00}
\acmDOI{10.1145/3391403.3399495}
\acmISBN{978-1-4503-7975-5/20/07}
\usepackage{booktabs} 
\usepackage[ruled]{algorithm2e} 

\SetAlFnt{\small}
\SetAlCapFnt{\small}
\SetAlCapNameFnt{\small}
\SetAlCapHSkip{0pt}
\IncMargin{-\parindent}

\else
\documentclass[format=acmsmall, review=false]{ec20acm}
\usepackage{arxiv}

\usepackage{booktabs} 

\usepackage[ruled]{algorithm2e} 

\SetAlFnt{\small}
\SetAlCapFnt{\small}
\SetAlCapNameFnt{\small}
\SetAlCapHSkip{0pt}
\IncMargin{-\parindent}
\fi

\usepackage{color}
\input{setup.tex}
\begin{document}
\title{Credible, Truthful, and Two-Round (Optimal) Auctions via Cryptographic Commitments}

\author{Matheus V. X. Ferreira}
\email{mvxf@cs.princeton.edu}
\affiliation{%
  \institution{Princeton University}
  \city{Princeton}
  \state{New Jersey}
}
\author{S. Matthew Weinberg}
\email{smweinberg@princeton.edu}
\affiliation{%
  \institution{Princeton University}
  \city{Princeton}
  \state{New Jersey}
}

\begin{abstract}
We consider the sale of a single item to multiple buyers by a revenue-maximizing seller. Recent work of Akbarpour and Li formalizes \emph{credibility} as an auction desideratum, and prove that the only optimal, credible, strategyproof auction is the ascending price auction with reserves~\cite{akbarpour2020credible}.

In contrast, when buyers' valuations are MHR, we show that the mild additional assumption of a cryptographically secure commitment scheme suffices for a simple \emph{two-round} auction which is optimal, strategyproof, and credible (even when the number of bidders is only known by the auctioneer).

We extend our analysis to the case when buyer valuations are $\alpha$-strongly regular for any $\alpha > 0$, up to arbitrary $\varepsilon$ in credibility. Interestingly, we also prove that this construction cannot be extended to regular distributions, nor can the $\varepsilon$ be removed with multiple bidders.

\end{abstract}
\maketitle

\section{Introduction}\label{sec:introduction}
We consider a revenue-maximizing auctioneer with a single item to sell to multiple bidders. Starting from Myerson's seminal work, it is traditionally assumed that the seller can commit to an auction format but that buyers must be incentivized to report their true values. Several recent works have moved beyond this assumption in repeated auctions, for example, where sellers can commit to a particular auction \emph{today}, but not to their behavior \emph{tomorrow} (e.g., ~\cite{devanur2014perfect, immorlica2017repeated, liu2019auctions}). Even more recent work of Akbarpour and Li proposes a framework also for one-shot auctions~\cite{akbarpour2020credible}. Our paper fits within this later framework.

Specifically, each buyer $i$ has value $v_i$ for the item, which is drawn independently from a distribution $D_i$, and the seller knows these distributions but not the precise values. As in~\cite{myerson1981optimal}, we seek auctions which are incentive compatible, and optimal among all incentive compatible auctions.~\cite{akbarpour2020credible} introduces a new desideratum, \emph{credibility}. Informally, an auction is credible if the auctioneer themselves is incentivized to execute the auction in earnest, even when permitted to cheat in ways that are undetectable to the bidders (see Section~\ref{sec:preliminaries} for a formal definition in single-item auctions).


Akbarpour and Li prove a comprehensive trilemma for single-item auctions: Myerson's auction is the unique truthful, one-round, revenue-maximizing auction, but it is not credible. Moreover, the ascending-price auction {(with reserves)} is the unique truthful, revenue-maximizing, credible auction, but it requires an unbounded number of rounds (that is, the number of rounds until termination cannot be bounded by any function only of the number of bidders. Even with two bidders, the number of rounds required is a function of the distributions). Finally, the first-price auction (with reserves) is the unique revenue-maximizing, credible, one-round auction, but it is not truthful. 

Classical auction theory might take truthfulness as a first-order concern, and view the tradeoff between bounded-round and credibility as a second-order concern. But as more and more auctions are run online, credibility is not just a ``bonus feature'', but a serious consideration. Specifically, reserve price-setting in ad auctions is often opaque, and a desire for transparency in execution has led major ad exchanges to switch from truthful second-price auctions to non-truthful (but credible) first-price auctions~\cite{klemperer2002really,sluis2019}. At the same time, these auctions are executed in milliseconds and must conclude before a search browser loads, so bounding the number of rounds is now a first-order concern as well. 

In this context, the trilemma of~\cite{akbarpour2020credible} may feel like a negative result: it is impossible to achieve all three first-order desiderata at once. Our main result circumvents their trilemma and provides a truthful, revenue-maximizing, credible, two-round auction, \emph{under the assumption of basic cryptographic primitives}. That is, viewed through the framework of~\cite{akbarpour2020credible} verbatim, our auctions are not credible (see Section~\ref{sec:dra} for an example). But, provably, no auctioneer can find a profitable deviation without breaking standard cryptographic assumptions. 

Interestingly, our construction is not a magic bullet with a trivial proof --- we must still carefully reason about the incentives of the auctioneer within our framework. Informally, our main results are (all under the assumption of a cryptographically secure commitment scheme, see Section~\ref{sec:preliminaries} for formal assumption, and also for formal definitions of distribution classes):
\begin{itemize}
\item When all $D_i$ are MHR, there is a truthful, revenue-maximizing, credible, two-round auction (Theorem~\ref{thm:mhr}).
\item When all $D_i$ are $\alpha$-strongly regular for any $\alpha \in (0,1)$, there is a truthful, revenue-maximizing, $\varepsilon$-credible, two-round auction (Theorem~\ref{thm:alphamhr}).
\item When there is a single bidder whose distribution is $\alpha$-strongly regular for any $\alpha \in (0,1)$, there is a truthful, revenue-maximizing, credible, two-round auction (Proposition~\ref{prop:alphaone}).
\item This auction is \emph{not} necessarily credible when there is a single buyer from a regular distribution, so extensions to regular distributions are not possible (Theorem~\ref{thm:regular}).
\item For any $\alpha \in (0,1)$, this auction is \emph{not} necessarily credible when all $D_i$ are $\alpha$-strongly regular, so the $\varepsilon$ is necessary in bullet two (Theorem~\ref{thm:alphamhrneg}).
\end{itemize}

\subsection{Brief Technical Overview}
Our auctions are still fairly simple and require only the basic cryptographic primitive of \emph{commitment schemes}. Informally, a commitment scheme allows a sender to send a \emph{commitment} $c_i$ to a bid $b_i$, such that any user who sees only $c_i$ learns absolutely nothing about $b_i$. Moreover, the sender can later \emph{reveal} $b_i$, in a way that proves they committed to $b_i$ in the first place (assuming the sender is computationally bounded). So our skeleton is simply to (a) ask each bidder to commit to their bid, (b) forward these commitments to all other bidders, (c) ask each bidder to reveal, (d) forward these revealed bids to all other bidders. We formalize this strawman auction in Section~\ref{sec:strawman}.

The outlook for this strawman auction initially looks promising: it is truthful, revenue-optimal, and two-round. Like in the auctions considered in~\cite{akbarpour2020credible}, the primary way in which the auctioneer can deviate is by submitting fake bids. It is not too hard to argue that \emph{if the auctioneer must reveal all committed bids}, then there is no way the auctioneer can deviate from being honest in a way that is both undetectable and profitable. However, the auction must have a well-defined execution even if some bids are \emph{concealed} (that is, the committed bid is never revealed). Should the auction simply stall? If so, it is undoubtedly in the auctioneer's interest to reveal all bids. Still, this auction is extremely not robust to latency, or an adversarial attack (simply commit a bid and disappear). Perhaps the auction should reboot? This also seems undesirable, as now the auctioneer has learned some private information. 

A natural suggestion (implemented in the strawman auction) is instead to replace all missing bids with $0$. This change, however, now gives the auctioneer a new class of potential deviations: they can commit to many different fake bids and reveal them selectively based on the true bids. It is not hard to see that this auction is not credible. 

We propose a straightforward modification, which is to fine any bidder who commits but does not reveal, \emph{and pay this fine to the winning bidder}. Now, the auctioneer faces a tradeoff: they can still commit to as many fake bids as they like, and they can still selectively reveal them. But for every bid they choose to conceal, they pay a fine. The entire technical portion of this paper is understanding when a sufficiently large fine exists to disincentivize the auctioneer from cheating in this particular way, and how large this fine must be. The bullet points above summarize our findings: such fines exist when all distributions are MHR, and (almost) exist when all distributions are $\alpha$-strongly regular, but do not necessarily exist even with one buyer from a regular distribution. 

\subsection{Related Work}
We have already overviewed the most related work above: we work in the model proposed by~\cite{akbarpour2020credible}, additionally with cryptographic primitives. There is a substantial literature generally on secure multi-party computation since Yao's millionaire problem~\cite{yao1982protocols} (see chapter 7 of \cite{goldreich2009foundations} for a survey on the topic), most of which is unrelated to our paper.

The easiest distinction between (most of) these works and ours is that they are not \emph{Sybil-proof}. Specifically, there is some trusted setup where every participant has an identity. Results such as~\cite{nurmi1993cryptographic} replace commitments with strong public-key infrastructure in our strawman proposal. Specifically, such protocols assume that a \emph{majority of participants} are honestly following the protocol. In online auctions, there is no hope of preventing the auctioneer from creating thousands of fake bidders if it will (undetectably) increase their revenue (so while a majority of ``real participants'' may be honest, the ``digital participants'' are nearly-unanimously \emph{not} following the protocol). A second distinction is that these protocols are often extremely complex, and certainly do not terminate in two rounds. Indeed, a central challenge for modern research in multi-party computation is developing practically reasonable protocols.

Our work is not the first to propose the use of fines to disincentivize participants from aborting a protocol~\cite{bradford2008protocol,bentov2014use}, as there are known impossibility results (without monetary incentives) when participants can abort~\cite{cleve1986limits}.

\subsection{Roadmap}
Section~\ref{sec:preliminaries} formalizes our problem of study, including cryptographic primitives. Section~\ref{sec:strawman} analyzes the strawman auction as a warmup. Section~\ref{sec:dra} proposes our auction, proves some basic facts, and states our main results. Sections~\ref{sec:regular} through~\ref{sec:strong-regularity} prove our main results. All technical sections present intuition, along with proofs (although some technical lemmas are deferred to the appendix). 

\section{Preliminaries}\label{sec:preliminaries}
We first overview formalities with regards to auctions. Our model and definitions are identical to~\cite{akbarpour2020credible}, but repeated for clarity and completeness.

\subsection{Auctions}
There is a single seller with a single indivisible item and $n$ bidders. Each bidder $i$ has a private value $v_i$ for the item, which is drawn from a distribution $D_i$. We let $D:= \times_i D_i$, and use $\rev(D)$ to denote the expected revenue of the optimal auction when buyers are drawn from $D$. To be concrete: each bidder $i$ knows only their value $v_i$ and type distribution $D_i$, but not $n$ or $D_{-i}$. The seller knows $n$ and $D$.\footnote{It is worth briefly noting that the Ascending Price Auction (and our auctions) are credible, even when the auctioneer can misrepresent $n$ and $D_{-i}$. On the other hand, the Second Price Auction is \emph{not} credible, even when all bidders know $n$ and $D_{-i}$.}

\vspace{1mm}\noindent\textbf{Communication and rounds.} The seller communicates with each bidder using a private channel (and this is the only communication --- the bidders do not communicate with each other). In every round, the following occurs: (a) each bidder chooses a message to send to the auctioneer, (b) the auctioneer processes all received messages, (c) the auctioneer chooses a (personalized) message to send to each bidder. At any point, the auctioneer may terminate and select a winner of the item (potentially no one), and charge prices. Importantly, each bidder communicates only with the auctioneer, and learns only whether or not they win the item and how much they pay upon termination (if they lose, they do not learn who wins, nor how much the winner pays). We also assume there is a default message $\bot$, which is sent if the bidder stays silent during a round.

\vspace{1mm}\noindent\textbf{Game.} Observe that the communication model induces an extended form game among the bidders and the auctioneer. Like~\cite{akbarpour2020credible}, we are interested in the case where the auctioneer commits publicly to a strategy which terminates in finite rounds with probability $1$. This induces an extended form game among the bidders. We'll refer to this game as the auction, and repeat the following definitions:

\begin{definition}[Strategyproof/Ex-Post Nash/Individually Rational] An auction is \emph{strategyproof} if {for all $n \in \mathbb N_{+}$, and for all $i \in [n]$}, there exists a mapping $s_i(\cdot)$ from values to strategies, and additionally for all $i$ and all $\vec{v}$, $s_i(v_i)$ is a best response of bidder $i$ to $\vec{s}_{-i}(\vec{v}_{-i})$. That is, for all $\vec{v}$, if buyers' valuations are $\vec{v}$, then $\langle s_1(v_1),\ldots, s_n(v_n)\rangle$ forms an ex-post Nash.\footnote{In other words, an auction need not be direct-revelation in order to be strategyproof, but there must exist a strategy ($s_i(v_i)$) which bidder $i$ can use which is akin to ``reporting $v_i$.''}

In this paper, we will only consider auctions for which there is a unique $s_1(\cdot),\ldots, s_n(\cdot)$ that always form an ex-post Nash, and refer to these strategies as ``telling the truth.''

An auction is \emph{individually rational} if telling the truth guarantees non-negative expected utility, ex-post.
\end{definition}

\begin{definition}[Safe Deviation] A \emph{safe deviation} for the auctioneer in the communication game is a strategy which does not necessarily implement the promised auction, but for every bidder $i$, their personal communication with the auctioneer and resulting allocation/price is consistent with some $(\vec{s}_{-i})^i := \langle s_1^i, s_2^i, ..., s_{i-1}^i, s_{i+1}^i, ..., s_{n_i}^i\rangle$ where bidder $i$ believes there are $n_i$ total bidders, and each $s_{j}^i$ is a valid strategy for the auction. Importantly, bidders can have inconsistent views of what is the communication game, depending on their interaction with the auctioneer.
\end{definition}

Observe quickly that if an auction is strategyproof, then the strategy of bidder $i$ is independent of $n_i$, the number of bidder they believe to be in the auction.

\begin{definition}[Credible] An auction is \emph{credible} if, in expectation over $\vec{v} \leftarrow D$, and conditioned on buyers being {\emph{truthful}}, executing the auction in earnest maximizes expected revenue over all safe deviations.
\end{definition}

\vspace{1mm}

\begin{example}[Second-Price Auction]~\cite{akbarpour2020credible} establishes that the second-price auction is not credible. Consider when $v_1 = 5$ and $v_2 = 10$. An earnest execution of the second-price auction would give the item to bidder $2$ and charge $5$. However, the auctioneer could instead give bidder $2$ the item and charge $9$ --- this is a safe deviation because it is consistent with buyer $1$ bidding $9$.
\end{example}

\subsection{Computational Assumptions and Basic Cryptography}
The {main} difference between our model and that of~\cite{akbarpour2020credible} is that we consider computationally-bounded participants and the existence of basic cryptographic primitives.

\vspace{1mm}\noindent\textbf{Commitment Scheme.} A commitment scheme is a function $\commit(\cdot,\cdot)$ which takes as input a message $m$, a one-time pad $r$, and outputs a commitment $c$. Informally, a scheme is computationally binding if a computationally-bounded seller cannot find an $m \neq m'$ and $r,r'$ such that $\commit(m,r) = \commit(m',r')$. A scheme is perfectly hiding if the distribution of commitments produced on message $m$ when $r$ is uniformly random \emph{is independent of $m$} (and therefore, even a computationally unbounded receiver learns nothing about $m$). 

\begin{assumption}\label{ass:commitment} There exists a cryptographic commitment scheme satisfying:
\begin{itemize}
\item (Efficiency) The function $\commit(\cdot,\cdot)$ can be implemented in time $\poly(|m|,|r|)$.
\item (Computationally Binding) For any algorithm $A$ which takes as input a length $k$, terminates in expected time $\poly(k)$, and outputs $m, r, m', r'$ with $|m|,|m'|,|r|,|r'| \leq \poly(k)$, $A$ breaks commitment w.p.~$\leq 2^{-\Omega(k)}$. Formally: $\Pr[\commit(m,r) = \commit(m',r'),\text{ and } m \neq m'] \leq 2^{-\Omega(k)}$.
\item (Perfectly Hiding) The distributions of $\commit(m,r)$ and $\commit(m',r')$, when $r$ and $r'$ are uniformly random, are identical distributions for all $m$ and $m'$. 
\item (Non-malleable) See~\cite{dolev2003nonmalleable,fischlin2000efficient} for a formal definition (as formal definitions are quite involved). Informally, imagine the following scenario: First, Alice sends $c:=\commit(m,r)$ to Bob. Then, Bob sends $c':=\commit(f(m), r')$ to Charlie \emph{without knowing $m,r$}. Then, Alice reveals $m,r$ to Bob. Then, Bob reveals $f(m), r'$ to Charlie. In this example, Bob has modified Alice's commitment to $m$ to instead be a commitment to $f(m)$. While Bob does not necessarily know $m$ when creating $c'$, he does know that $c'$ is a commitment to $f(m)$ (whatever $m$ happens to be), and he does know how to reveal $f(m),r'$ once he learns $m,r$. If Bob can perform this style ``man-in-the-middle'' attack (in poly-time) for a non-identity function $f(\cdot)$, then $\commit(\cdot,\cdot)$ is malleable. Informally, $\commit(\cdot,\cdot)$ is non-malleable if the above scenario is feasible only when $f(\cdot)$ is the identity function (but this is not a formal definition, see~\cite{dolev2003nonmalleable, fischlin2000efficient}).
\end{itemize}
\end{assumption}

There are indeed commitment schemes which are believed to satisfy Assumption~\ref{ass:commitment}, such as the Pedersen scheme with digital signatures.\footnote{Briefly, the Pedersen scheme requires a group of prime order $p$ under which the discrete logarithm is (believed to be) hard, with generator $g$. Every potential receiver of a message raises $g$ to a random power to get another generator $h$, and publicly announces $h$. Then $\commit(m,r):=g^m\cdot h^r$. Observe that for all $c$ and all $m$, there exists a unique $r$ such that $g^m \cdot h^r = c$ (so the scheme is perfectly hiding). But if a sender were able to break their commitment, this would explicitly learn $\log_g(h)$, so it is also computationally binding. As stated, the scheme is malleable: an adversary could see $g^m h^r$ and multiply it by $g^2$ to now get $g^{m+2}h^r = \commit(m+2,r)$. The scheme can be made non-malleable by first using any non-malleable digital signature scheme. Note that to use exactly the Pedersen commitment scheme (with digital signatures), every bidder $i$ would need to share their own $h_i$ in order to receive binding commitments (and a public key), which can be done in one additional preprocessing round, and this preprocessing round could be done once and reused across indefinitely-many auctions.} Note that the particular choice of a perfectly hiding (versus computationally hiding) scheme is not crucial for the spirit of our results. However, it does allow significantly cleaner theorem statements. Similarly, our main positive result (Theorem~\ref{thm:mhr}) doesn't require non-malleability, although it does make proofs cleaner (our main extension, Theorem~\ref{thm:alphamhr} necessarily requires non-malleability). Informally, Assumption~\ref{ass:commitment} implies that unless the auctioneer is relying on events which occur with probability at most $2^{-\Omega(k)}$, or is computationally unbounded, the auctioneer cannot perform an \emph{unreasonable deviation}, defined below.

\begin{definition}[Reasonable Deviation] Say that a commitment $c$ is \emph{explicitly tied to} $(m,r)$ if the participant (bidder or auctioneer) who created $c$ explicitly computed $c:=\commit(m,r)$. Note that this implies a commitment $c$ is explicitly tied to at most one $(m, r)$. A \emph{reasonable deviation} for the auctioneer in the communication game is a strategy such that whenever the auctioneer reveals a commitment to $c$, with $c=\commit(m,r)$, $c$ was explicitly tied to $(m,r)$.
\end{definition}

Observe that one kind of unreasonable deviation would violate computational binding: the auctioneer might compute $c:=\commit(m,r)$, but later reveal that $c:=\commit(m',r')$ (unreasonable because $c$ is explicitly tied to $(m,r)$, not $(m',r')$). Another kind would violate non-malleability: the auctioneer might receive commitments $c_1:=\commit(m_1,r_1), c_2:=\commit(m_2,r_2)$ and send $c_3:=\commit(\max\{m_1,m_2\},r_1+r_2)$ without knowing $m_1,m_2$ (unreasonable because $c_3$ is not explicitly tied to anything).

\begin{definition}[Computationally Credible] An auction is \emph{computationally credible} if, in expectation over $\vec{v} \leftarrow D$, and buyers being truthful, the auctioneer maximizes their expected revenue, over all deviations which are both safe and reasonable, by executing the auction in earnest.

An auction is computationally $\varepsilon$-credible if executing the auction in earnest yields a $(1-\varepsilon)$-fraction of the expected revenue of any safe, reasonable deviation.
\end{definition}

Our main results will design auctions that are computationally credible (Theorem~\ref{thm:mhr}). One of our extensions will design an auction which is computationally $\varepsilon$-credible (Theorem~\ref{thm:alphamhr}), and some of our lower bounds rule out $\varepsilon$-credible mechanisms for $\varepsilon$ arbitrarily close to one (Theorem~\ref{thm:regular}).

\subsection{Virtual Values}
For a continuous single-dimensional distribution with CDF $F$ and PDF $f$, the \emph{virtual value} of $x$ is $\varphi^F(x) := x - \frac{1 - F(x)}{f(x)}$. We also use $h^F(x) := \frac{f(x)}{1 - F(x)}$ the \emph{hazard rate} of $F$. We drop the superscript $F$ if it is clear from context, and will use subscripts of $i$ instead of superscripts of $D_i$ (e.g. $h_i(x):=h^{D_i}(x)$). Seminal work of Myerson asserts that the expected revenue of any strategyproof mechanism is its expected virtual welfare.

\begin{theorem}[\cite{myerson1981optimal}]\label{thm:myerson}
Let a strategy proof mechanism award bidder $i$ the item with probability $x_i(\vec{b})$ on bids $\vec{b}$, and charge them $p_i(\vec{b})$. Then:
\begin{equation*}
E_{\vec{v} \leftarrow D}\left[\sum_{i = 1}^n p_i(\vec{v}) \right] = E_{\vec{v} \leftarrow D}\left[\sum_{i = 1}^n x_i(\vec{v})\varphi_i(v_i)\right]
\end{equation*}
\end{theorem}

Finally, we conclude with a definition of classes of distributions which are relevant for our results.

\begin{definition}[Regular, MHR, $\alpha$-Strongly Regular] A distribution $F$ is \emph{$\alpha$-strongly regular} if for all $v' \geq v$, $\varphi^F(v') - \varphi^F(v) \geq \alpha (v'-v)$. A distribution is \emph{regular} if it is $0$-strongly regular, and \emph{monotone hazard rate (MHR)} if it is $1$-strongly regular.
\end{definition}

\section{Strawman Computationally Credible Auctions}\label{sec:strawman}
We propose a simple modification to any direct revelation mechanism, which turns these one-round mechanisms into two-round mechanisms. In round one, the buyer's communication is simply a commitment to a bid. The auctioneer's communication is to forward these commitments to all bidders. In round two, the buyer's communication is to decommit (reveal their bid to the auctioneer). The auctioneer's communication is to forward all (decommitted) bids to the buyers. We use the terminology \emph{reveal} $c_i$ when a message $(m_i,r_i)$ such that $\commit(m_i,r_i) = c_i$ is sent, and \emph{conceal} $c_i$ when some other pair is sent instead.

\begin{definition}[Strawman Auction] Let $\commit(\cdot,\cdot)$ be a commitment scheme satisfying Assumption~\ref{ass:commitment}. For a given direct revelation mechanism, with allocation rule $\vec{x}(\cdot)$ and payment rule $\vec{p}(\cdot)$, $\strawman(\vec{x},\vec{p})$ is the following auction:\\

\noindent$1^{st}$ Round:
\begin{itemize}
    \item Each bidder $i$ picks a bid, $b_i$, draws $r_i$ uniformly at random, and sends $c_i:=\commit(b_i,r_i)$.
    \item The auctioneer sends each commitment to all buyers.
\end{itemize}
$2^{nd}$ Round:
\begin{itemize}
    \item Each bidder $i$ sends $(b_i,r_i)$ to the auctioneer.
    \item The auctioneer forwards each $(b_i,r_i)$ to all buyers.
\end{itemize}
Resolution:
\begin{itemize}
    \item Let $S$ denote the set of bidders for which $c_i=\commit(b_i,r_i)$, and let $b'_i := b_i \cdot I(i \in S)$. Allocate and charge payments according to $\vec{x}(\vec{b}'), \vec{p}(\vec{b}')$.
\end{itemize}
\end{definition}

In particular, observe that the auction's behavior must be well-defined even when not all commitments are revealed. We quickly observe that the Strawman Auction preserves incentive compatibility:

\begin{observation}\label{obs:strawman} Let $(\vec{x},\vec{p})$ be a strategyproof, individually rational, direct revelation mechanism. Then $\strawman(\vec{x},\vec{p})$ is also strategyproof and individually rational. In particular, it is an ex-post Nash for each bidder to set $b_i := v_i$, and to reveal in round two.
\end{observation}
\begin{proof}
Because $(\vec{x},\vec{p})$ is individually rational, no bidder can benefit by replacing their bid with $0$ by concealing in round two. Given that bidder $i$ will reveal, and that all other bidders will also reveal,\footnote{Note that this is necessary: it is not a \emph{dominant strategy} to be honest. Bidder $1$ could use a weird strategy ``Commit to $c_1:=(\infty,0)$. If I am sent a commitment of $c=\commit(5,12)$, then conceal. Otherwise, reveal.'' If bidder $1$ uses this strategy (and you are the only other bidder), it is a better response to just send $\commit(5,12)$ and reveal, rather than being honest.} it is best for bidder $i$ to commit to $v_i$ (because $(\vec{x},\vec{p})$ is strategyproof).
\end{proof}

Observation~\ref{obs:strawman} establishes that this modification preserves strategyproofness. One might hope that it also encourages the auctioneer to behave honestly (if $\vec{x},\vec{p}$ is the revenue-optimal auction) because they do not know any of the buyers' bids before round two. So while the auctioneer can create fake bidders and submit fake bids, it seems like these bids may simply act as a reserve. And indeed, \emph{if the auctioneer must reveal all fake bids}, the only reasonable deviations are to reveal the precise fake bids selected in round one (which was chosen with no information about buyers' values). Therefore, the Strawman optimal auction would be computationally credible by the same reasoning used in~\cite{akbarpour2020credible} for the ascending price auction. 

Unfortunately, the auction's behavior must be well-defined even when some bids are concealed, and the auction cannot merely stall. For example, bidders may naturally drop out between rounds due to latency issues, or attackers may adversarially bid and conceal to stall the system. Restarting the auction is perhaps even worse, as now the auctioneer has learned some private information from those bidders who did participate honestly. This means that while it is not a safe, reasonable deviation for the auctioneer to change their commitment, it is indeed a safe, reasonable deviation for the auctioneer to simply conceal some fake bids. The following example establishes that such deviations violate computational credibility of the Strawman auction.

\begin{example}[\strawman\ is not computationally credible] Consider that there is a single (real) buyer, whose value is drawn uniformly from $\{1,2\}$, and consider the Strawman second-price auction with reserve $1$, which tie-breaks lexicographically. The auctioneer will get expected revenue $1$ by being honest (which is optimal among all strategyproof auctions). Instead, they could create a fake bidder, and always commit to $b_2 = 2$. After bidder $1$ reveals in round $2$, the auctioneer can either (a) reveal $b_2$, if $b_1 = 2$, causing $b'_2=2$ and yielding revenue $2$ or (b) conceal, if $b_1 = 1$, causing $b'_2 = 0$ and yielding revenue $1$. This gets the auctioneer expected revenue $3/2$.
\end{example}

The main takeaway from this section is that the \strawman\ auction serves as a good base for a computationally credible, strategyproof auction; however, we cannot force the auctioneer (or any bidder, for that matter) to reveal their bids. Our solution in Section~\ref{sec:dra} is to instead fine all bidders (including fake bidders) who conceal, to disincentivize this particular safe, reasonable deviation.

\section{Deferred Revelation Auction}\label{sec:dra}
In this section, we describe the Deferred Revelation Auction (DRA) and prove basic facts that will be useful throughout all of our analyses. Rather than state the auction as a reduction, we directly apply it to Myerson's revenue-optimal auction. Below, recall that $\bar{\varphi}(\cdot)$ is Myerson's ironed virtual value function, which is the upper concave envelope of $\varphi(\cdot)$ (for further details, see~\cite{myerson1981optimal,hartline2013mechanism}). Recall also that, by~\cite{myerson1981optimal}, the allocation rule of the revenue-optimal single-item auction is to award the item to the buyer with the highest non-negative ironed virtual value (tie-breaking lexicographically).\footnote{When we say ``tie-break lexicographically'', we mean ``break all ties in favor of the lowest-indexed bidder''.}

\begin{definition}[Deferred Revelation Auction] Let $\commit(\cdot,\cdot)$ be a commitment scheme satisfying Assumption~\ref{ass:commitment}. For a given \emph{fine function} $f(\cdot,\cdot)$, $\dra(f)$ is the following auction:\\

\noindent$1^{st}$ Round:
\begin{itemize}
    \item Each buyer $i$ picks a bid, $b_i$, draws a one-time pad $r_i$ uniformly at random, and sends $c_i:=\commit(b_i,r_i)$. The distribution $D_i$ from which $v_i$ is drawn is known to the auctioneer.
    \item The auctioneer sends ($c_i,D_i,i$) to all buyers. Let $n_i$ denote the number of tuples sent to buyer $i$ (including their own).
\end{itemize}
$2^{nd}$ Round:
\begin{itemize}
    \item Each buyer $i$ sends $(b_i,r_i)$ to the auctioneer.
    \item The auctioneer forwards each $(b_i,r_i)$ to all buyers.
\end{itemize}
Resolution:
\begin{itemize}
    \item Let $S$ denote the set of buyers for which $c_i=\commit(b_i,r_i)$, and let $b'_i := b_i \cdot I(i \in S)$. Let $i^*:= \arg\max_{i \in S} \{\bar{\varphi}_i(b_i)\}$.
\item If $\bar{\varphi}_{i^*}(b_{i^*}) > 0$, award buyer $i^*$ the item. Charge them $\bar{\varphi}^{-1}_{i^*}(\max\{0,\max_{i \in S \setminus\{i^*\}}\{\bar{\varphi}_i(b_i)\}\})$.\footnote{Here, we define the inverse of a monotone function $g(\cdot)$ to be $g^{-1}(y)=\inf_x \{x|\ g(x) \geq y\}$.}
\item Additionally, all $i \notin S$ pay buyer $i^*$ a fine equal to $f(n_{i^*},D_{i^*})$.
\end{itemize}
Tie-breaking:
\begin{itemize}
    \item All ties are broken lexicographically, with the auctioneer treated as ``buyer zero''. With this, we will write all inequalities as $>$ or $<$, taking this tie-breaking already into account. 
\end{itemize}
\end{definition}

Above, we are essentially running the optimal \strawman\ auction, but fining any buyers who conceal \emph{and paying these fines to the winning buyer}. Intuitively, this helps in the following way: during round one, the auctioneer can certainly gamble and commit to several fake bids. However, they run the risk of accidentally overshooting the winning bid. In the \strawman\ auction, they could simply conceal these bids. In $\dra(f)$, they must instead pay some fine $f(n_{i^*},D_{i^*})$. Intuitively, it seems that if the fine is large enough, the auctioneer would rather just be honest than commit to any fake bids and run the risk of paying a huge fine. This turns out to be true when each $D_i$ is MHR, but not in general. Before stating our main results, we recap the safe, reasonable deviations.

\begin{enumerate}
\item The seller might create fake buyers during round one.
\item The seller may selectively choose which commitments to send to buyer $i$.
\begin{itemize}
\item The seller might not send $c_j$ at all.
\item The seller might send $c_j$, but with some $D'_j$ instead of the true $D_j$.
\item If $\commit(\cdot,\cdot)$ were malleable, the seller could apply some function $g(\cdot)$ to $b_j$ and forward instead $c'_j=\commit(g(b_j),r_j)$. We assumed in Assumption~\ref{ass:commitment} that $\commit(\cdot,\cdot)$ is non-malleable to avoid this, although Theorem~\ref{thm:mhr} holds even without this assumption.
\item All of these might depend on $\vec{b}_{-i}$, but not $b_i$.\footnote{For example, the seller might solicit a commitment from all buyers. Then, in increasing order of $i$, they could forward some commitments to buyer $i$, and ask $i$ to reveal. Then, after terminating this for all buyers, they could go back and reveal commitments. As far as any individual buyer can tell, the timeline appears correct based on their interaction with the seller.}
\end{itemize}
\item The seller might conceal a commitment. This decision can depend on the entire $\vec{b}$.
\end{enumerate}

Before reasoning about computational credibility, we quickly observe that $\dra(f)$ is indeed strategyproof and optimal.

\begin{observation}\label{obs:sp} For all $f$, $\dra(f)$ is strategyproof and revenue-optimal.
\end{observation}
\begin{proof}
$\dra(f)$ is clearly optimal, as it simply runs Myerson's auction. It is also strategyproof: because Myerson's auction is individually rational, buyers have no incentive to conceal their commitment in round two. Given that all buyers will reveal their commitments, it is in buyer $i$'s interest to commit to $v_i$, because Myerson's auction is truthful.
\end{proof}

It is more challenging to reason when $\dra(f)$ is computationally credible. While there are many ways the seller might deviate, our approach to upper bounding the seller's revenue, fortunately, boils down to one vector of parameters determined by the seller's decisions during round one.

\begin{definition} For a triple $(c_j = \commit(b_j,r_j), D_j,j)$ sent to bidder $i$, denote the \emph{effective bid} by $\beta_{ij}:=\bar{\varphi}_i^{-1}(\bar{\varphi}_j(b_j))$. We call the \emph{effective commitment} to buyer $i$ as $\beta_i:=\max\{\bar{\varphi}_i^{-1}(0),\max_j \{\beta_{ij}\}\}$. We call the \emph{effective reveal} to buyer $i$ as $\gamma_i:=\max\{\bar{\varphi}_i^{-1}(0),\max_{j,c_j \text{ is revealed to $i$}}\{\beta_{ij}\}\}$. 

Recall that $\beta_i$ is a function only of $\vec{b}_{-i}$, and $\gamma_i$ is a function of $\beta_i$ and $\vec{b}$.
\end{definition}

\begin{observation}\label{obs:easy} For all $f$, under any safe, reasonable deviation to $\dra(f)$, bidder $i$ receives the item if and only if $v_i > \gamma_i$. Therefore, $v_i > \gamma_i$ for at most one bidder.
\end{observation}

We now quickly show that Observation~\ref{obs:easy} is enough to show that $\dra(f)$ is computationally credible whenever each $D_i$ is bounded. 

\begin{observation}\label{obs:bounded} Let each $D_i$ be bounded, and let $f(n,D_i):=\inf\{x : \Pr_{v \leftarrow D_i}[v > x] = 0\}+1$. Then $\dra(f)$ is optimal, strategyproof, and computationally credible for the instance $D = \times_i D_i$.
\end{observation}

\begin{proof}
Optimality and strategyproofness follow directly from Observation~\ref{obs:sp}. To see that $\dra(f)$ is computationally credible, observe that the auctioneer will certainly get negative revenue if they ever conceal a fake bid sent to buyer $i$ and sell them the item (because buyer $i$ will pay at most $\inf\{x : \Pr_{v \leftarrow D_i}[v > x] = 0\}$, while the auctioneer must pay a strictly larger fine). Therefore, any optimal safe, reasonable deviation has $v_i < \gamma_i \Leftrightarrow v_i < \beta_i$. Indeed, the $\Leftarrow$ implication is trivial, as $\gamma_i \leq \beta_i$. The $\Rightarrow$ implication follows because the auctioneer would get negative revenue selling the item to buyer $i$, and could strictly improve their revenue by just revealing all commitments and not selling the item. 

Once we have this implication, observe that buyer $i$ now wins the item if and only if $v_i > \beta_i$, and $\beta_i$ is a function of $\vec{b}_{-i}$. Moreover, when buyer $i$ wins, they will pay $\beta_i$. This is a truthful mechanism, and therefore it achieves revenue no better than Myerson's optimal auction. To recap: we have shown that every safe, reasonable deviation is strictly outperformed by another deviation, which implements a truthful mechanism, which is outperformed by executing the auction in earnest.
\end{proof}

Observation~\ref{obs:bounded} illustrates one idea to reason about computational credibility of $\dra(f)$, but does not really shed much insight, as it is essentially just forcing the auctioneer to reveal all commitments. Our main results, therefore, concern unbounded distributions (or significantly shrinking the fines necessary for bounded distributions), where there do not exist sufficiently large fines to trivially force the auctioneer to always reveal. We begin with our positive results. Below, $r(D_0)$ denotes the Myerson reserve for a one-dimensional distribution $D_0$.

\begin{theorem}\label{thm:mhr} Let $f(n,D_i):=r(D_i)$.\footnote{Note that when $D_i$ is MHR, $r(D_i) = \Theta(\rev(D_i))$~\cite{cai2011extreme}.} Then when all $D_i$ are MHR (bounded or unbounded), $\dra(f)$ is optimal, strategyproof, and computationally credible.
\end{theorem}

\begin{theorem}\label{thm:alphamhr} For all $\varepsilon,\alpha > 0$, there exists an $f(\cdot,\cdot)$ such that $f(n,D_i) \leq \poly_\alpha(n,r(D_i),1/\varepsilon)$ for all $n,D_i$,\footnote{By the notation $\poly_\alpha(\cdot)$, we mean that for all fixed $\alpha$, the fine is $\poly(n,D(r_i),1/\varepsilon)$. The precise fine which we prove suffices is $f(n,D_i):=\left(\frac{2n^2}{\varepsilon \alpha}\right)^{\frac{1-\alpha}{\alpha}}\cdot (1-\alpha)^{-1/\alpha} \cdot r(D_i)$.} such that when all $D_i$ are unbounded and $\alpha$-strongly regular, $\dra(f)$ is optimal, strategyproof, and computationally $\varepsilon$-credible.
\end{theorem}

Theorem~\ref{thm:alphamhr} can be improved to remove the $\varepsilon$ when there is just a single bidder, but not otherwise (see Theorem~\ref{thm:alphamhrneg} shortly after).

\begin{proposition}\label{prop:alphaone}
Moreover, for all $\alpha > 0$, there exists an $f(\cdot,\cdot)$ with $f(n,D_0):=\Theta_\alpha(r(D_0))$ for all $n$, such that when $D_0$ is $\alpha$-strongly regular, $\dra(f)$ is optimal, strategyproof, and computationally credible when there is a single (real) buyer from $D_0$.\footnote{The precise fine which we prove suffices is $f(n,D):= r(D)\bigg(\bigg(\frac{1}{1-\alpha}\bigg)^{\frac{1}{1-\alpha}}\frac{1}{\alpha}\bigg)^{\frac{1-\alpha}{\alpha}}$.}
\end{proposition}

Theorem~\ref{thm:mhr} is our main positive result: it asserts that there is a reasonably-sized fine, which depends only on $D_i$ and not even on $n$, such that these fines are sufficient to deter the auctioneer from submitting fake bids. Proposition~\ref{prop:alphaone} extends Theorem~\ref{thm:mhr} to $\alpha$-strongly regular distributions when there is just a single (real) bidder. Theorem~\ref{thm:alphamhr} is an extension to multiple bidders, but is a relaxation in two ways: the mechanism is only $\varepsilon$-credible, and the fine now depends on $n$. Our main negative results establish that these are necessary, and Theorem~\ref{thm:alphamhr} is essentially the limit of what $\dra(f)$ achieves within the framework of $\alpha$-strongly regular distributions. Our negative results are as follows:

\begin{theorem}\label{thm:regular} There exists an unbounded regular distribution $D_0$, such that for all $f(\cdot,\cdot)$, $\dra(f)$ is \emph{not} computationally $\varepsilon$-credible for the instance $D_0$ and \emph{any} $\varepsilon < 1$.
\end{theorem}

\begin{theorem}\label{thm:alphamhrneg}
For all $f(\cdot,\cdot)$, all $\alpha < 1$, and all $n>1$, there exists an unbounded $D_0$ that is $\alpha$-strongly regular such that $\dra(f)$ is \emph{not} computationally credible for the instance $D :=\times_{i=1}^n D_0$.
\end{theorem}

Before continuing, let us parse the results, which clearly distinguish between MHR, $\alpha$-strongly MHR, and regular distributions. On one extreme, \dra\ works as well as could be hoped for when all distributions are MHR: there is a fine which is \emph{independent of the number of buyers} which suffices to ensure that \dra\ is computationally credible. On the other extreme, \dra\ does not work well at all for arbitrary regular distributions: even when $n=1$, there may not exist a sufficiently large fine to discourage the auctioneer from cheating, and cheating may yield unboundedly more revenue than honesty. In the middle, we see that Theorem~\ref{thm:alphamhr} does not distinguish between different values of $\alpha \in (0,1)$. In this range, positive results are possible, but not quite so strong as for MHR distributions. Moreover, the positive results we prove {for $\dra(f)$} are tight.

We conclude this section by revisiting our simple example under \dra\ instead of \strawman. Section~\ref{sec:regular} follows immediately afterwards, and proves Theorem~\ref{thm:regular} (perhaps unsurprisingly, the witness $D_0$ is the equal-revenue curve). This will give an intuition for the technical challenges, and why stronger assumptions are necessary to have the positive results in Theorems~\ref{thm:mhr} and~\ref{thm:alphamhr}, whose proofs follow in Sections~\ref{sec:mhr} and~\ref{sec:strong-regularity}. 

\begin{example} Consider that there is a single (real) buyer, whose value is drawn from $D_1$, which is the uniform distribution on $\{1,2\}$. Let also $f(n,D_1):=1$ for all $n$. Consider now the auction $\dra(f)$. The auctioneer will get expected revenue $1$ by being honest and not submitting any fake bids (which is optimal among all strategyproof auctions). Instead, the auctioneer could submit any number of fake bids. It is clear that it only makes sense to submit fake bids of $2$, and also that it is unnecessary to submit multiple fake bids of the same value. 

In order to be a reasonable deviation, if the auctioneer submits a fake bid of $b_2 = 2$, then after buyer $1$ reveals in round $2$, the auctioneer can either reveal $b_2 = 2$, or conceal. In order to be a safe deviation, the auctioneer must set a price of $b_2$ to buyer $1$ if they reveal, and set a price of $1$ otherwise. In particular, observe that while the auctioneer can guarantee revenue $2$ when $b_1=2$ (by revealing), the best revenue they can guarantee when $b_1=1$ is $0$. If they reveal, then they pay no fines but also receive no payment. If they conceal, then they get payment of $1$, but also pay a fine of $1$, for a net payment of $0$. Therefore, no matter what strategy the auctioneer uses, they get revenue at most $1$ in expectation, the same as being honest.

Observe that if we only consider \emph{safe} (but unreasonable) deviations, then the auctioneer could commit to $b_2 = 2$, but reveal instead a commitment to $b_2 = 1$ when $b_1 = 1$. Of course, doing so would require breaking the cryptographic commitment scheme, an event that can be made less likely than the inverse number of atoms in the universe. So this mechanism is not \emph{credible}, but only \emph{computationally credible}, and this example highlights the distinction. 
\end{example}

\section{Example: DRA on Regular Distributions}\label{sec:regular}
In this section, we prove Theorem~\ref{thm:regular}. The main intuition is that the equal-revenue curve is so heavy-tailed that no matter how big the fines are, there are always some sufficiently-high fake bids that the auctioneer can set to extract additional revenue while barely ever paying the fine.

\begin{proof}[Proof of Theorem~\ref{thm:regular}] Let $D_0$ denote the equal-revenue distribution, which has CDF $1-1/x$ on $[1,\infty)$. The optimal revenue that the seller can achieve by earnestly running a truthful auction for one bidder drawn from $D_0$ is $1$. Consider now any fine function $f(\cdot,\cdot)$, and simply refer to $L_n:=f(n+1,D_0)$ as the fine the seller must pay per hidden fake bid, if they submit $n$ fake bids. We show in fact that for all $f(\cdot,\cdot)$, there not only exists a safe, reasonable deviation which achieves revenue $>1$, but also one that achieves revenue $> r$ \emph{for any $r$}.

For a given $r$, let $n\geq r+2$. Consider now the following construction of fake bids: Set $b_i:=n^{2i}\cdot L_n$ for all $i \in [n]$. For simplicity of notation, define $b_{n+1}:=\infty$. The seller's strategy is then:
\begin{itemize}
\item Commit to a bid $b_i$ for all $i \in [n]$.
\item When the bidder's bid $b$ is revealed:
\begin{itemize}
\item If $b < b_1$, reveal all bids.
\item Otherwise, if $b \in [b_i,b_{i+1})$, reveal bids $b_1,\ldots, b_i$, and conceal $b_{i+1},\ldots, b_n$. 
\end{itemize}
\end{itemize}

We now want to compute the seller's expected revenue for this strategy. We do this by first upper bounding the total expected \emph{fines} that the seller will pay. 

\begin{claim}\label{claim:reg1} The total expected fines paid by the seller in expectation is at most $1/n$. 
\end{claim}
\begin{proof}
Observe that the seller only ever pays a fine when $b > b_1$. Because $b$ is drawn from the equal-revenue curve, this occurs with probability at most $1/b_1 = 1/(n^2L_n)$. Moreover, the seller submits only $n$ fake bids, and therefore the total fine they pay, conditioned on paying a fine at all, is at most $n L_n$. Therefore, the total expected fines paid by the seller in expectation is at most $1/n$.
\end{proof}

Next, we show that the seller's expected payment received by the buyer is still large.
\begin{claim}\label{claim:reg2} The expected revenue that the seller receives is at least $n-1/n$.
\end{claim}

\begin{proof}
We compute the probability that the buyer pays exactly $b_i$, for all $i$. Observe that the buyer pays exactly $b_i$ whenever $b\in [b_i, b_{i+1}]$ which occurs with probability exactly $1/b_i - 1/b_{i+1}$, because $b$ is drawn from the equal-revenue curve. As $b_{i+1} = n^2 b_i$, this probability is exactly $(1-1/n^2)/b_i$ (or this is a lower bound, when $i = n$). Therefore, the expected revenue can be written as:
$$\sum_{i=1}^n b_i \cdot \Pr[\text{buyer pays exactly $b_i$}] \geq \sum_i b_i \cdot (1-1/n^2)/b_i = n-1/n.$$ 
\end{proof}

Claims~\ref{claim:reg1} and~\ref{claim:reg2} together establish that the seller achieves expected revenue at least $n-2/n \geq r$, as desired.
\end{proof}

The key feature of the equal-revenue curve which drives the proof of Theorem~\ref{thm:regular} is that for all probabilities $p$, there exists an optimal reserve which is exceeded with probability at most $p$. This allowed us to set extremely high ``reserves'', to get revenue as if we are setting each of these reserves independently, while also paying fines so extremely rarely that it barely matters. In Appendix~\ref{app:regular}, we show that it is really a condition like this which drives Theorem~\ref{thm:regular}, and not just that the equal-revenue curve has infinite expectation (by providing an example of a distribution with infinite expectation and a choice of $f$ for which $\dra(f)$ is computationally $\varepsilon$-credible for that distribution). We will also try to use this as intuition when explaining our (more technical) proofs for the MHR and $\alpha$-strongly regular cases.

\section{DRA is Credible for MHR Distributions}\label{sec:mhr}
In this section, we consider the performance of $\dra$ on MHR distributions. Drawing intuition from what drove the proof of Theorem~\ref{thm:regular}, the key feature which enables a strong positive result for (even unbounded) MHR distributions is that the revenue generated by reserves significantly above the optimal reserve shrinks exponentially fast. 

Recall from Section~\ref{sec:dra} that $\beta_i$ denotes the effective commitment to buyer $i$, and that it is a function of $\vec{b}_{-i}$. Our analysis breaks down the expected revenue achieved by the seller using any round one strategy into two terms: revenue from cases where there exists an $i$ such that $b_i > \beta_i$, and revenue when all $i$ satisfy $b_i < \beta_i$. The first case, which we proceed with now, has similarities to the analysis in Section~\ref{sec:dra}, but is more precise so that it can be combined with the second case. 

\begin{lemma}\label{lem:onebig}
For any $D$, $f$, consider any strategy of the seller, which is a safe, reasonable deviation to $\dra(f)$. Let $R(\vec{b})$ denote the revenue achieved by the seller on bids $\vec{b}$ using this strategy. Then:

$$\mathbb{E}_{\vec{v} \leftarrow D}[R(\vec{v}) \cdot I(\exists\ i, v_i > \beta_i)] \leq \mathbb{E}_{\vec{v} \leftarrow D}[\max_j \{\varphi_j(v_j)\} \cdot I(\exists\ i, v_i > \beta_i)].$$
\end{lemma}
\begin{proof}
Observe first that by Observation~\ref{obs:easy}, whenever there exists an $i$ such that $v_i > \beta_i$, there is a unique such $i$ (otherwise, each such $i$ certainly satisfies $v_i > \gamma_i$ as $\gamma_i \leq \beta_i$, which contradicts Observation~\ref{obs:easy}). So consider the allocation rule which awards the item to bidder $i$ if and only if $v_i > \beta_i$, and charges them $\beta_i$. Observe first that the expected revenue of this allocation rule is \emph{at least} $\mathbb{E}_{\vec{v} \leftarrow D}[R(\vec{v}) \cdot I(\exists\ i, v_i > \beta_i)]$. Indeed, if the seller chooses to reveal all commitments sent to bidder $i$, then this will be exactly the expected revenue. If the seller (sub-optimally) chooses instead to conceal some commitments, they simply pay additional fines and get less revenue.

Importantly, observe also that this allocation rule is monotone, as $\beta_i$ doesn't depend on $v_i$. Moreover, observe that this allocation/payment rule is truthful. Therefore, Myerson's Lemma implies that its expected revenue is exactly its expected virtual surplus, and its expected virtual surplus is exactly $\mathbb{E}_{\vec{v} \leftarrow D}[\sum_i \varphi_i(v_i) \cdot I(v_i > \beta_i)]$. This gives the following chain of inequalities:

\begin{align*}
\mathbb{E}_{\vec{v} \leftarrow D}\left[R(\vec{v}) \cdot I(\exists\ i, v_i > \beta_i)\right]&\leq \mathbb{E}_{\vec{v} \leftarrow D}\left[\sum_i \beta_i\cdot I(v_i > \beta_i)\right]\\
&= \mathbb{E}_{\vec{v} \leftarrow D}\left[\sum_i \varphi_i(v_i) \cdot I(v_i > \beta_i)\right]\\
 &\leq \mathbb{E}_{\vec{v} \leftarrow D}\left[\max_j \{\varphi_j(v_j)\} \cdot I(\exists\ i, v_i > \beta_i)\right].
\end{align*}

The first line follows from the reasoning in the first paragraph: $\dra(f)$ will charge bidder $i$ at most $\beta_i$ when they win. The second line is just Myerson's lemma. The final line is just upper bounding a particular virtual value with the maximum virtual value and uses Observation~\ref{obs:easy} to conclude that no more than one indicator variable in the sum can be non-zero.
\end{proof}

The second step is now to bound the optimal revenue the seller can get from cases where $v_i < \beta_i$ for all $i$. The following technical lemma will be a crucial step in this part of the analysis. Intuitively, Lemma~\ref{lem:mhr} states that the expected \emph{value} of a draw from an MHR distribution, conditioned on being large, is not much more than its expected \emph{virtual value} under the same conditioning. Below, recall that we defined $r(D_0)$ to be the Myerson reserve of $D_0$.

\begin{lemma}\label{lem:mhr}
Let $D_0$ be MHR. Let $E$ be any event such that $\Pr_{v \leftarrow D_0}[v \geq r(D_0)|E] = 1$. Then:
$$\mathbb{E}_{v \leftarrow D_0}[v| E] \leq \mathbb{E}_{v \leftarrow D_0}[\varphi(v)|E] + r(D_0).$$

Equivalently, $\mathbb{E}_{v \leftarrow D_0}[v\cdot I(E)] \leq \mathbb{E}_{v \leftarrow D_0}[\varphi(v)\cdot I(E)] + r(D_0)\cdot \Pr[E]$.
\end{lemma}

\begin{proof}
Recall that because $D_0$ is MHR, and $v \geq r(D_0)$, we have that $\varphi(v) - \varphi(r(D_0)) \geq v - r(D_0)$ whenever event $E$ occurs. Recalling that $\varphi(r(D_0)) = 0$ by definition, this rearranges to $v \leq r(D_0) + \varphi(v)$. We then immediately conclude:
\begin{align*}
\mathbb{E}_{v \leftarrow D_0}[v| E] &\leq \mathbb{E}_{v \leftarrow D_0}[\varphi(v)+ r(D_0)|E]\\
&=  \mathbb{E}_{v \leftarrow D_0}[\varphi(v)|E] + r(D_0).
\end{align*}
\end{proof}

\begin{corollary}\label{cor:nonebig} Let each $D_i$ be MHR, and consider $\dra(f)$ where $f(n,D_i)=r(D_i)$ for all $i$. Consider any strategy of the seller, which is a safe, reasonable deviation, and let $R(\vec{b})$ denote the revenue achieved by the seller on bids $\vec{b}$ using this strategy. Then:
$$\mathbb{E}_{\vec{v} \leftarrow D}[R(\vec{v}) \cdot I(\forall\ i, v_i < \beta_i)] \leq \mathbb{E}_{\vec{v} \leftarrow D}[\max\{0,\max_j \{\varphi_j(v_j)\}\} \cdot I(\forall\ i, v_i < \beta_i)].$$
\end{corollary}
\begin{proof}
For ease of notation let $r_i:= r(D_i)$, and let $X_i(\vec{v})$ denote the indicator random variable for the event that $v_i > r_i$, $v_j < \beta_j$ for all $j$, and the item is awarded to bidder $i$. Then clearly when this occurs, the payment made by bidder $i$ is at most $v_i$. But additionally, selling the item to buyer $i$ when $v_i < \beta_i$ requires concealing at least one commitment and paying a fine (otherwise, bidder $i$ expects not to win the item). As the fine charged per concealed commitment is $r_i$, this means that the seller's total revenue is at most $v_i - r_i$. In particular, this also concludes that the seller's total revenue when awarding the item to buyer $i$ when $v_i < r_i$ is non-positive. Therefore, we can write:

$$\mathbb{E}_{\vec{v} \leftarrow D}[R(\vec{v}) \cdot I(\forall\ i, v_i < \beta_i)] \leq \mathbb{E}_{\vec{v} \leftarrow D}\left[\sum_i (v_i - r_i)\cdot X_i(\vec{v})\right].$$

But now let's consider $\mathbb{E}_{\vec{v} \leftarrow D}[ (v_i - r_i) \cdot X_i(\vec{v})]$ separately for each $i$. The event $X_i(\vec{v})=1$ satisfies the hypotheses of Lemma~\ref{lem:mhr}, as $X_i(\vec{v})=1$ implies that $v_i > r_i$. Therefore, Lemma~\ref{lem:mhr} allows us to conclude that:

\begin{align*}
\mathbb{E}_{\vec{v} \leftarrow D}[(v_i - r_i)\cdot X_i(\vec{v})] &=  \mathbb{E}_{\vec{v} \leftarrow D}\left[v_i \cdot X_i(\vec{v})] - r_i \cdot \Pr[X_i(\vec{v}) =1\right]\\
& \leq \mathbb{E}_{\vec{v} \leftarrow D}\left[\varphi_i(v_i) \cdot X_i(\vec{v})\right].
\end{align*}

The first line is just linearity of expectation, and the second line follows by Lemma~\ref{lem:mhr}. Now, we can put everything together to conclude:
\begin{align*}
\mathbb{E}_{\vec{v} \leftarrow D}[R(\vec{v}) \cdot I(\forall\ i, v_i < \beta_i)] &\leq \mathbb{E}_{\vec{v} \leftarrow D}\left[\sum_i \varphi_i(v_i) \cdot X_i(\vec{v})\right]\\
&\leq\mathbb{E}_{\vec{v} \leftarrow D}\left[\max_i \{\varphi_i(v_i)\} \cdot \left(\sum_i X_i(\vec{v})\right)\right]\\
&\leq \mathbb{E}_{\vec{v} \leftarrow D}\left[\max\{0,\max_j \{\varphi_j(v_j)\}\} \cdot I(\forall\ i, v_i < \beta_i)\right].
\end{align*}
The first line is simply restating the work above. The second line is just upper bounding each $\varphi_i(v_i)$ with the maximum virtual value. The final line simply observes that at most one of the indicators $X_i(\vec{v})$ can be one (because at most one bidder can receive the item), and that a prerequisite for any of them to be one is that all $v_i < \beta_i$. 
\end{proof}

Lemma~\ref{lem:onebig} and Corollary~\ref{cor:nonebig} together suffice to prove Theorem~\ref{thm:mhr}.

\begin{proof}[Proof of Theorem~\ref{thm:mhr}]
Lemma~\ref{lem:onebig} upper bounds the expected revenue of any safe, reasonable deviation when some $v_i > \beta_i$. Corollary~\ref{cor:nonebig} upper bounds the expected revenue of any safe, reasonable deviation when all $v_i < \beta_i$. Together, this implies that for any safe, reasonable deviation:
\begin{align*}
\mathbb{E}_{\vec{v} \leftarrow D}[R(\vec{v})] &=\mathbb{E}_{\vec{v} \leftarrow D}[R(\vec{v}) \cdot I(\exists\ i, v_i > \beta_i)] + \mathbb{E}_{\vec{v} \leftarrow D}[R(\vec{v}) \cdot I(\forall\ i, v_i < \beta_i)]\\
&\leq \mathbb{E}_{\vec{v} \leftarrow D}[\max_j \{\varphi_j(v_j)\} \cdot I(\exists\ i, v_i > \beta_i)] \\
&\qquad+\mathbb{E}_{\vec{v} \leftarrow D}[\max\{0,\max_j \{\varphi_j(v_j)\}\} \cdot I(\forall\ i, v_i < \beta_i)]\\
&=\mathbb{E}_{\vec{v} \leftarrow D}[\max\{0,\max_j \{\varphi_j(v_j)\}\}].
\end{align*}

The RHS is now precisely the expected revenue that the seller achieves by executing the protocol in earnest, so this series of inequalities explicitly witnesses that every safe, reasonable deviation yields expected revenue at most that of being honest.
\end{proof}

To repeat the key steps in the proof: Lemma~\ref{lem:onebig} doesn't use at all the particular form of $f(\cdot,\cdot)$, nor that each $D_i$ is MHR. It merely says that the revenue achieved from cases where the seller may as well reveal all commitments is the same as a truthful auction (because these commitments to $i$ are a function only of $\vec{b}_{-i}$). Corollary~\ref{cor:nonebig} uses the particular form of $f(\cdot,\cdot)$ and that each $D_i$ is MHR to conclude that even when the seller might strategically conceal some commitments, it does no better than a truthful auction. {Interestingly, observe that the entire proof only used the property that $\beta_i$ can be written as a function of $\vec{b}_{-i}$, which is true even when the commitment scheme is malleable. So Theorem~\ref{thm:mhr} holds even for malleable commitment schemes (but still requires the commitment scheme to be binding).}

\section{Extensions and Limitations of $\alpha$-Strongly Regular Distributions}\label{sec:strong-regularity}
We now provide an extension of Theorem~\ref{thm:mhr} to $\alpha$-strongly regular distributions, but also prove the limits of such an extension. The proof of our extension follows a similar outline to Section~\ref{sec:mhr}. In particular, recall that Lemma~\ref{lem:onebig} held \emph{for all distributions}, not just MHR. So we will use Lemma~\ref{lem:onebig} verbatim to handle the case where some $v_i > \beta_i$. Lemma~\ref{lem:mhr}, however, requires the MHR assumption. Our first step is to extend (and relax) Lemma~\ref{lem:mhr}. The proof of Lemma~\ref{lem:alphamhr} and Corollary~\ref{cor:alphanonebig} are similar to Section~\ref{sec:mhr}, and deferred to Appendix~\ref{app:alpha}.

\begin{lemma}\label{lem:alphamhr}
Let $D_0$ be $\alpha$-strongly regular. Let $E$ be such that $\Pr_{v \leftarrow D_0}\left[v \geq r(D_0)|E\right] = 1$. Then:
$$\mathbb{E}_{v \leftarrow D_0}\left[v| E\right] \leq \frac{1}{\alpha} \cdot \mathbb{E}_{v \leftarrow D_0}\left[\varphi(v)|E\right] + r(D_0).$$

Equivalently, $\mathbb{E}_{v \leftarrow D_0}\left[v\cdot I(E)\right] \leq \frac{1}{\alpha} \cdot \mathbb{E}_{v \leftarrow D_0}\left[\varphi(v)\cdot I(E)\right] + r(D_0)\cdot \Pr\left[E\right]$.
\end{lemma}

In Corollary~\ref{cor:alphanonebig} below, we will consider again safe, reasonable deviations from a particular $\dra(f)$. Below, we'll let $R(\vec{b})$ denote the revenue achieved by the seller (using this particular deviation) on bids $\vec{b}$, and $k_i:= f(n_i,D_i)$.

\begin{corollary}\label{cor:alphanonebig} Let each $D_i$ be $\alpha$-strongly regular, and consider $\dra(f)$ where $f(n,D_i) \geq r(D_i)$ for all $n,i$. Consider any strategy of the seller which is a safe, reasonable deviation. Finally, let $X_i(\vec{v})$ denote the indicator random variable for the event that the item is awarded to bidder $i$, $v_i > k_i$, and $v_j < \beta_j$ for all $j$. Then:
$$\mathbb{E}_{\vec{v} \leftarrow D}\left[R(\vec{v}) \cdot I(\forall\ i, v_i < \beta_i)\right] \leq \sum_i \mathbb{E}_{\vec{v} \leftarrow D}\left[(\varphi_i(v_i)/\alpha+r_i-k_i)\cdot X_i(\vec{v})\right].$$
\end{corollary}

From here, making use of Corollary~\ref{cor:alphanonebig} is not as straight-forward as in the MHR case. We first need another technical lemma, bounding the achievable revenue by posting a very high price for a single $\alpha$-strongly regular distribution. The proof of Lemma~\ref{lem:alphapricing} appears in Appendix~\ref{app:alpha}.\footnote{Lemma~\ref{lem:alpha-prob-tail-2} gives a stronger bound when $\alpha \rightarrow 1$, but the simpler bound in Lemma~\ref{lem:alphapricing} suffices for all of our results.}

\begin{lemma}\label{lem:alphapricing} Let $D_0$ be $\alpha$-strongly regular. Then for all $p \geq r(D_0)$,
$$p \cdot \Pr_{v \leftarrow D_0}[v \geq p] \leq r(D_0) \cdot \Pr_{v \leftarrow D_0}[v \geq r(D_0)] \cdot (1-\alpha)^{-1/(1-\alpha)}\cdot \bigg(\frac{r(D_0)}{p}\bigg)^{\frac{\alpha}{1-\alpha}} $$
\end{lemma}

And finally, we need one more technical lemma before we can wrap up the proof of Theorem~\ref{thm:alphamhr}. This technical lemma is the only reason why Theorem~\ref{thm:alphamhr} applies to unbounded distributions, and also the only reason why we need non-malleability of the commitment scheme. We show in Appendix~\ref{app:malleable} that both non-malleability and unbounded distributions are indeed necessary (via a counterexample to Theorem~\ref{thm:alphamhr} otherwise). The proof of Lemma~\ref{lem:all} is included, as it has no counterpart in Section~\ref{sec:mhr}.

\begin{lemma}\label{lem:all} Let each $D_i$ be unbounded. Then for all $f$, all $j \in [n]$, and any safe, reasonable deviation in $\dra(f)$ it must be that for at least $j$ distinct bidders, $n_i \geq n-j+1$. In particular, $\sum_i 1/n_i^2 \leq \pi^2/6 \leq 2$.
\end{lemma}
\begin{proof}
For simplicity of notation, relabel the bidders $1,\ldots, n$ by the order in which the seller requests their decommitment (if some are requested simultaneously, break those ties arbitrarily). Importantly, observe that the seller cannot request decommitment from a bidder until they have forwarded all commitments. Therefore, the decision of which commitments to forward to bidder $i$ can depend only on $\vec{b}_{<i}$.

So now assume for contradiction that the lemma fails for some $j$. Then there is some bidder $i \leq j$ with $n_i \leq n-i$ (recall that $n_i$ includes their own commitment too). In particular, this means that there is some bidder $\ell > i$ whose commitment was not forwarded to bidder $i$, and that $b_\ell$ was completely unknown when this decision was made. In particular, the following situation now has non-zero probability:
\begin{itemize}
\item First, draw $v_{-i,\ell}$ to determine which commitments to forward to bidder $i$. Observe that this also suffices to define $\beta_i$, as it is independent of both $v_i$ and $v_\ell$.
\item Now, it is entirely possible that $v_i > \beta_i$, as $D_i$ is unbounded. Observe that determining this only requires additionally drawing $v_i$.
\item Now, this sets $\beta_\ell$. As we have yet to draw $v_\ell$, it is entirely possible that $v_\ell > \beta_\ell$, as $D_\ell$ is unbounded.
\end{itemize}

The above derives a contradiction to the deviation being safe and reasonable, as now two distinct buyers are both expecting to win the item. The ``In particular,\ldots'' part of the statement follows simply as the sum is {maximized} when there is exactly one bidder with $n_i = j$, for all $j \in [n]$. 
\end{proof}

We can now wrap up the proof of Theorem~\ref{thm:alphamhr}. 

\begin{proof}[Proof of Theorem~\ref{thm:alphamhr}]
Consider first combining Lemma~\ref{lem:onebig} and Corollary~\ref{cor:alphanonebig}. If we set $f(n,D_i):=\left(\frac{2n^2}{\varepsilon \alpha}\right)^{\frac{1-\alpha}{\alpha}}\cdot (1-\alpha)^{-1/\alpha} \cdot r(D_i)$, we get:
\begin{align*}
\mathbb{E}_{\vec{v} \leftarrow D}[R(\vec{v})] &=\mathbb{E}_{\vec{v} \leftarrow D}[R(\vec{v}) \cdot I(\exists\ i, v_i > \beta_i)] + \mathbb{E}_{\vec{v} \leftarrow D}[R(\vec{v}) \cdot I(\forall\ i, v_i < \beta_i)]\\
&\leq \mathbb{E}_{\vec{v} \leftarrow D}[\max_j \{\varphi_j(v_j)\} \cdot I(\exists\ i, v_i > \beta_i)] +\sum_i \mathbb{E}_{\vec{v} \leftarrow D}\left[(\varphi_i(v_i)/\alpha+r(D_i)-k_i)\cdot X_i(\vec{v})\right]\\
&\leq\mathbb{E}_{\vec{v} \leftarrow D}[\max\{0,\max_j \{\varphi_j(v_j)\}\}] + \sum_i \mathbb{E}_{v_i \leftarrow D_i}\left[\varphi_i(v_i)\cdot (1/\alpha)\cdot I(v_i > k_i)\right]\\
&= \mathbb{E}_{\vec{v} \leftarrow D}[\max\{0,\max_j \{\varphi_j(v_j)\}\}] + \sum_i k_i \cdot \Pr_{v_i \leftarrow D_i}[v_i > k_i]/\alpha\\
&\leq  \mathbb{E}_{\vec{v} \leftarrow D}[\max\{0,\max_j \{\varphi_j(v_j)\}\}]\\
&\qquad + \sum_i (1-\alpha)^{-1/(1-\alpha)}\cdot (k_i/r(D_i))^{-\frac{\alpha}{1-\alpha}} \cdot r(D_i) \cdot \Pr_{v_i \leftarrow D_i}[v_i \geq r(D_i)]/\alpha\\
&\leq  \mathbb{E}_{\vec{v} \leftarrow D}[\max\{0,\max_j \{\varphi_j(v_j)\}\}] + \sum_i \frac{\varepsilon}{2n_i^2} \cdot r(D_i) \cdot \Pr_{v_i \leftarrow D_i}[v_i \geq r(D_i)]\\
&\leq \rev(D) + \varepsilon \rev(D) \cdot \sum_i \frac{1}{2n_i^2}\\
&\leq (1+\varepsilon)\rev(D).
\end{align*}

The first line is just linearity of expectation. The second line is Lemma~\ref{lem:onebig} and Corollary~\ref{cor:alphanonebig}. The third line simply observes that $X_i(\vec{v}) = 1 \Rightarrow v_i > k_i$, and also that $k_i > r(D_i)$. The fourth simply observes that the right-hand term of line three is the expected virtual welfare of an auction (which sells the item to bidder $i$ whenever $v_i > k_i$), and the right-hand term of line four is the expected revenue of that same auction (so they are equal by Myerson's lemma). The fifth is a direct application of Lemma~\ref{lem:alphapricing}. The sixth uses our particular choice of $f(n,D_i)$. The seventh simply observes that both $\mathbb{E}_{\vec{v} \leftarrow D}[\max\{0,\max_j \{\varphi_j(v_j)\}\}] = \rev(D)$, and also $r(D_i) \cdot \Pr_{v_i \leftarrow D_i}[v_i \geq r(D_i)] \leq \rev(D)$ (because this is just the revenue of selling only to bidder $i$). The final line follows directly from Lemma~\ref{lem:all}.
\end{proof}

{It may seem odd that Theorem~\ref{thm:alphamhr} requires both that the distributions are unbounded, and also that the commitment scheme is non-malleable (given that neither assumption is necessary for Theorem~\ref{thm:mhr}). Both of these assumptions show up only in the proof of Lemma~\ref{lem:all}, where we show that the auctioneer must send many commitments to each bidder. This perhaps seems like a technical artifact of the current proof approach, but surprisingly we show that both assumptions are necessary. Specifically, Theorem~\ref{thm:alphamhr} \emph{does not} hold when the commitment scheme is malleable, nor when the distributions are bounded. This establishes that there is (perhaps surprisingly) something integral about Lemma~\ref{lem:all} to the proof of Theorem~\ref{thm:alphamhr}. See Appendix~\ref{app:malleable} for formal theorem statements and proofs.}

A full proof of Proposition~\ref{prop:alphaone} appears in Appendix~\ref{app:alphaone}, which reuses several technical lemmas.

\subsection{Limits of \dra\ for $\alpha$-Strongly Regular}
We conclude by establishing that the $\varepsilon$ in Theorem~\ref{thm:alphamhr} cannot be improved for $n>1$ bidders (whereas Proposition~\ref{prop:alphaone} removes it for $n=1$ bidder). We provide a complete proof below, which will give further intuition for why the $\varepsilon$ is not needed for MHR distributions, or only a single bidder.

\begin{proof}[Proof of Theorem~\ref{thm:alphamhrneg}]
For any $n > 1$, our safe, reasonable strategy will do the following. First, it will always interact honestly with bidders $\neq 1$. Also, it will \emph{almost always} also interact honestly with bidder $1$. In the extremely rare case that the maximum bid from bidders $\neq 1$ is unusually high, then it will try to cheat bidder $1$. To be clear, the auctioneer's strategy will do the following (for simplicity of notation in what follows, we denote by $k:=f(n+1,D_0)$):

\begin{enumerate}
\item Honestly solicit commitments from all $n$ bidders.
\item Honestly forward all commitments to all bidders $\neq 1$, and ask for them to reveal. Let $j^*:=\arg\max_{j \neq 1}\{b_j\}$.
\item If $b_{j^*} \leq T$, for some threshold $T$ to be set later, honestly forward all commitments to bidder $1$ as well, and ask bidder $1$ to reveal. Execute the auction honestly.
\item If instead $b_{j^*} > T$, forward instead all commitments to bidder $1$, along with one fake commitment to $b=b_{j^*}+k$. 
\begin{itemize}
\item If $v_1 \leq T$, reveal $b$ and execute the auction honestly. 
\item If $v_1 \geq b$, reveal $b$ and execute the auction honestly.
\item If $v_1 \in (T, b)$, conceal $b$ and sell the item to bidder one.
\end{itemize}
\end{enumerate}

Observe that this is indeed a safe, reasonable deviation. From the perspective of each bidder, they first send a commitment, then receive commitments, then reveal their commitment, then learn which commitments are revealed/concealed. Intuitively, our proof will show that no matter how large $f(n+1, D_0)$ is, there is always a sufficiently large $T$ such that this fine becomes negligible and the increased revenue from setting a slightly higher ``reserve'' of $b$ becomes worth it (note, however, that this phenomenon does \emph{not} occur for MHR distributions, by Theorem~\ref{thm:mhr}). 

Our distribution $D_0$ will have the following CDF and PDF:

\begin{minipage}{.5\linewidth}
$$F^\alpha(v) = \begin{cases}0 \quad & \text{, $v < 1$}\\
1 - \big(\frac 1 v\big)^{\frac 1 {1-\alpha}} \quad & \text{, $v \geq 1$} \end{cases}$$
\end{minipage}
\begin{minipage}{.5\linewidth}
$$f^\alpha(v) = \begin{cases} 0 \quad &\text{, $v < 1$}\\
\frac{1}{1-\alpha}\big(\frac 1 v\big)^{\frac {2-\alpha}{1-\alpha}} \quad &\text{, $v \geq 1$}\end{cases}$$
\end{minipage}

The hazard rate of $F^\alpha$ is $h^{F^\alpha}(v) = \frac{1}{(1-\alpha) v}$ for $v \geq 1$ and the virtual value function of $F^\alpha$ is $\varphi^{F^\alpha}(v) = v - \frac{1}{h^{F^\alpha}(v)} = \alpha v$, so $D_0$ is $\alpha$-strongly regular. 

Observe that for our particular deviation, the revenue of the honest execution and our deviation differ \emph{only when $v_1 > v_{j^*} > T$}. The tradeoff the auctioneer chooses is that when $v_1 > v_{j^*}+k$, they get an additional revenue of $k$. But if instead $v_1 \in (v_{j^*},v_{j^*}+k)$, they have to pay a fine of $k$. Observe that the difference in both cases is exactly $k$, one in favor of cheating, and the other in favor of being honest. So we just want to check how big $v_{j^*}$ needs to be in order to have $\Pr[v_1 > v_{j^*}+k] > \Pr[v_1 \in (v_{j^*},v_{j^*}+k)]$. Again, observe that Theorem~\ref{thm:mhr} establishes that no such $v_{j^*}$ exists when $D_0$ is MHR and $k = r(D_0)$. But slightly relaxing this condition to $\alpha$-strongly regular for $\alpha < 1$ now implies the existence of such a $v_{j^*}$ for any $k$.

We upper bound the probability that $v_1 \in (v_{j^*}, v_{j^*} + k)$, conditioned on $v_1 \geq v_{j^*}$ (holds for any $v_{j^*} \geq 1$):
\begin{align*}
    \Pr[v_1 \in (v_{j^*}, v_{j^*} + k) | v_1 \geq v_{j^*}] &= \int_{v_{j^*}}^{v_{j^*} + k} \frac{f^{\alpha}(v)}{1 - F^{\alpha}(v_{j^*})} dv\\
    &\leq \int_{v_{j^*}}^{v_{j^*} + k} \frac{f^{\alpha}(v_{j^*})}{1 - F^{\alpha}(v_{j^*})} dv\\
    &= k \cdot h^{F^{\alpha}}(v_{j^*}) = \frac{k}{(1-\alpha)v_{j^*}}
\end{align*}

Observe this also implies $\Pr[v_1 \geq v_{j^*}+k | v_1 \geq v_{j^*}] = 1-\Pr[v_1 \in (v_{j^*}, v_{j^*} + k) | v_1 \geq v_{j^*}] \geq 1-\frac{k}{(1-\alpha)v_{j^*}}$. Together, these two claims immediately imply that (after multiplying both by $\Pr[v_1 \geq v_{j^*}]$):

$$v_{j^*}> \frac{2k}{1-\alpha}\Rightarrow \Pr[v_1 \geq v_{j^*}+k]> \Pr[v_1 \in (v_{j^*}, v_{j^*} + k)].$$

By the work above, this proves that this deviation is strictly profitable for any $T > \frac{2k}{1-\alpha}$.
\end{proof}

{In Appendix~\ref{app:malleable}, we consider \dra\ with a malleable commitment scheme. In particular, we show that if the commitment scheme is sufficiently malleable, then $\dra(f)$ is not $(1-\alpha-\varepsilon)$-credible for any $\varepsilon > 0$ for multiple $\alpha$-strongly regular buyers (this means that non-malleability is truly necessary for Theorem~\ref{thm:alphamhr}, and not a technical artifact of going through Lemma~\ref{lem:all}). We also show that this is tight, and that there exists an $f$ such that $\dra(f)$ is $(1-\alpha)$-credible for multiple $\alpha$-strongly regular buyers, even when the commitment scheme is malleable (but still computationally binding).}

\bibliographystyle{alpha}
\bibliography{masterbib}

\begin{thebibliography}{BPRP08}

\bibitem[AL20]{akbarpour2020credible}
Mohammad Akbarpour and Shengwu Li.
\newblock Credible auctions: A trilemma.
\newblock {\em Econometrica}, 88(2):425--467, 2020.

\bibitem[BK14]{bentov2014use}
Iddo Bentov and Ranjit Kumaresan.
\newblock How to use bitcoin to design fair protocols.
\newblock In {\em Annual Cryptology Conference}, pages 421--439. Springer,
  2014.

\bibitem[BPRP08]{bradford2008protocol}
Phillip~G Bradford, Sunju Park, Michael~H Rothkopf, and Heejin Park.
\newblock Protocol completion incentive problems in cryptographic vickrey
  auctions.
\newblock {\em Electronic Commerce Research}, 8(1-2):57--77, 2008.

\bibitem[CD11]{cai2011extreme}
Yang Cai and Constantinos Daskalakis.
\newblock Extreme-value theorems for optimal multidimensional pricing.
\newblock In {\em 2011 IEEE 52nd Annual Symposium on Foundations of Computer
  Science}, pages 522--531. IEEE, 2011.

\bibitem[Cle86]{cleve1986limits}
Richard Cleve.
\newblock Limits on the security of coin flips when half the processors are
  faulty.
\newblock In {\em Proceedings of the eighteenth annual ACM symposium on Theory
  of computing}, pages 364--369. ACM, 1986.

\bibitem[CR14]{cole2014sample}
Richard Cole and Tim Roughgarden.
\newblock The sample complexity of revenue maximization.
\newblock In {\em Proceedings of the forty-sixth annual ACM symposium on Theory
  of computing}, pages 243--252. ACM, 2014.

\bibitem[DDN03]{dolev2003nonmalleable}
Danny Dolev, Cynthia Dwork, and Moni Naor.
\newblock Nonmalleable cryptography.
\newblock {\em SIAM review}, 45(4):727--784, 2003.

\bibitem[DPS14]{devanur2014perfect}
Nikhil~R Devanur, Yuval Peres, and Balasubramanian Sivan.
\newblock Perfect bayesian equilibria in repeated sales.
\newblock In {\em Proceedings of the twenty-sixth annual ACM-SIAM symposium on
  Discrete algorithms}, pages 983--1002. SIAM, 2014.

\bibitem[FF00]{fischlin2000efficient}
Marc Fischlin and Roger Fischlin.
\newblock Efficient non-malleable commitment schemes.
\newblock In {\em Annual International Cryptology Conference}, pages 413--431.
  Springer, 2000.

\bibitem[Gol09]{goldreich2009foundations}
Oded Goldreich.
\newblock {\em Foundations of cryptography: volume 2, basic applications}.
\newblock Cambridge university press, 2009.

\bibitem[Har13]{hartline2013mechanism}
Jason~D Hartline.
\newblock Mechanism design and approximation.
\newblock {\em Book draft. October}, 122, 2013.

\bibitem[Hoo06]{hooshmand2006ultra}
MH~Hooshmand.
\newblock Ultra power and ultra exponential functions.
\newblock {\em Integral Transforms and Special Functions}, 17(8):549--558,
  2006.

\bibitem[ILPT17]{immorlica2017repeated}
Nicole Immorlica, Brendan Lucier, Emmanouil Pountourakis, and Samuel Taggart.
\newblock Repeated sales with multiple strategic buyers.
\newblock In {\em Proceedings of the 2017 ACM Conference on Economics and
  Computation}, pages 167--168, 2017.

\bibitem[Kle02]{klemperer2002really}
Paul Klemperer.
\newblock What really matters in auction design.
\newblock {\em Journal of economic perspectives}, 16(1):169--189, 2002.

\bibitem[LMSZ19]{liu2019auctions}
Qingmin Liu, Konrad Mierendorff, Xianwen Shi, and Weijie Zhong.
\newblock Auctions with limited commitment.
\newblock {\em American Economic Review}, 109(3):876--910, 2019.

\bibitem[Mye81]{myerson1981optimal}
Roger~B Myerson.
\newblock Optimal auction design.
\newblock {\em Mathematics of operations research}, 6(1):58--73, 1981.

\bibitem[NS93]{nurmi1993cryptographic}
Hannu Nurmi and Arto Salomaa.
\newblock Cryptographic protocols for vickrey auctions.
\newblock {\em Group Decision and Negotiation}, 2(4):363--373, 1993.

\bibitem[Slu19]{sluis2019}
Sarah Sluis.
\newblock Google switches to first-price auction, Mar 2019.

\bibitem[Yao82]{yao1982protocols}
Andrew Chi-Chih Yao.
\newblock Protocols for secure computations.
\newblock In {\em FOCS}, volume~82, pages 160--164, 1982.

\end{thebibliography}
\appendix
\section{Facts about $\alpha$-Strongly Regular Distributions}\label{sec:alphamhrfacts}
Here we provide a proof of the main technical lemma for $\alpha$-strongly regular distributions, which itself follows from two short structural lemmas. Lemma~\ref{lem:hazard-rate-bound} follows directly from the definition of virtual values and $\alpha$-strongly regular distributions \cite{cole2014sample}.

\begin{lemma}\label{lem:hazard-rate-bound}
If an $\alpha$-strongly regular distribution has CDF $F$ and PDF $f$, then for all $v' \geq v$, 
\begin{equation}\label{eq:hazard-lower-bound}
h(v') \geq \frac{1}{(1-\alpha)(v' - v) + 1/h(v)}
\end{equation}
\end{lemma}
\begin{proof}
For all $v' \geq v$, if $h(v)$ is the hazard rate of $F$, then $\varphi(v') = 1 - 1/h(v)$. By definition of $\alpha$-strongly regularity,
\begin{align*}
    &\varphi(v') - \varphi(v) = v' - 1/h(v') - v + 1/h(v) \geq \alpha(v' - v)\\
    &\implies 1/h(v') \leq (1 - \alpha)(v' - v) + 1/h(v)
\end{align*}
The latter implies the statement.
\end{proof}

\begin{lemma}\label{lem:alpha-prob-tail}
Let an $\alpha$-strongly regular distribution have CDF $F$ and PDF $f$, and let $r:=\varphi^{-1}(0)$. Then for all $x \geq r$,
$$Pr_{v \leftarrow D}[v \geq x] \leq Pr_{v \leftarrow D}[v \geq r]\cdot \bigg(\frac{r}{(1-\alpha)x + \alpha r}\bigg)^\frac{1}{1-\alpha}$$
\end{lemma}
\begin{proof}
Let $H(v) = \int_{0}^v h(x)dx$. A well-known property of hazard rates is that $1 - F(v) = e^{-H(v)}$. To see this, observe that $\frac{d}{dx}\ln(1-F(x)) = -\frac{f(x)}{1 - F(x)} = -h(x)$. By the fundamental theorem of calculus, $\int_0^v -h(x) dx = \ln(1 - F(v)) - \ln(1 - F(0)) = \ln(1 - F(v))$, which implies $1 - F(v) = e^{-\int_0^v h(x)dx} = e^{-H(v)}$.

By Lemma~\ref{lem:hazard-rate-bound}, we have
\begin{align*}
    H(v) &= \int_0^v h(x)dx = \int_{0}^r h(x) dx + \int_r^v h(x) dx\\
    &\geq H(r) + \int_r^v \frac{1}{(1-\alpha)(x - r) + r} dx\\
    &= H(r) + \frac{1}{1-\alpha}\big[\ln((1-\alpha)(x - r) + r)\big]^v_r\\
    &= H(r) + \frac{1}{1-\alpha}\ln\bigg(\frac{(1-\alpha) v + \alpha r}{r}\bigg)
\end{align*}
This implies that:
\begin{align*}
    Pr_{v \leftarrow D}[v \geq x] &= e^{-H(x)}\\
    &\leq e^{-H(r)}e^{\frac{1}{1-\alpha}\ln \frac {r} {\alpha r + (1-\alpha)x}}\\
    &= Pr_{v \leftarrow D}[v \geq r]\cdot \bigg(\frac{r}{\alpha r + (1-\alpha) x}\bigg)^{\frac{1}{1 - \alpha}}
\end{align*}
\end{proof}

When $(1-\alpha) = o(r/x)$, we obtain an exponential tail bound:
\begin{lemma}\label{lem:alpha-prob-tail-2}
Let an $\alpha$-strongly regular distribution have CDF $F$ and PDF $f$, and let $r:=\varphi^{-1}(0)$. Then for all $x \geq r$,
$$Pr_{v \leftarrow D}[v \geq x] \leq Pr_{v \leftarrow D}[v \geq r]  \cdot \exp\bigg(-\frac{x/r-1}{\alpha + (1-\alpha)x/r}\bigg)$$
\end{lemma}
\begin{proof}
We simply rewrite the second term in the inequality of Lemma~\ref{lem:alpha-prob-tail}:
\begin{align*}
\frac{1}{1-\alpha}\ln\bigg(\frac{r}{\alpha r + (1-\alpha) x}\bigg) &= -\frac{1}{1-\alpha}\ln(\alpha + (1-\alpha) x/r)\\
&\leq -\frac{1}{1-\alpha} \frac{\alpha + (1-\alpha)x/r  -1}{\alpha + (1-\alpha)x/r}\\
&= -\frac{1}{1-\alpha} \frac{(1-\alpha)(x/r-1)}{\alpha + (1-\alpha)x/r}\\
&= - \frac{x/r-1}{\alpha + (1-\alpha)x/r}
\end{align*}
The inequality uses the fact $\ln(x) \geq (x - 1)/x$. To conclude the proof, we simply exponentiate both sides of the inequality.
\end{proof}

\section{Omitted Proofs from Section~\ref{sec:strong-regularity}}\label{app:alpha}
\begin{proof}[Proof of Lemma~\ref{lem:alphamhr}]
Again recall that because $D_0$ is $\alpha$-strongly regular, and $v \geq r(D_0)$ whenever event $E$ occurs, that $\varphi(v) - \varphi(r(D_0)) \geq \alpha v - \alpha r(D_0)$. Recalling that $\varphi(r(D_0)) = 0$ by definition, this rearranges to $v \leq r(D_0) + \varphi(v)/\alpha$. We then immediately conclude:
\begin{align*}
\mathbb{E}_{v \leftarrow D_0}\left[v| E\right] &\leq \mathbb{E}_{v \leftarrow D_0}\left[\varphi(v)/\alpha+ r(D_0)|E\right]\\
&=  \frac{1}{\alpha} \cdot \mathbb{E}_{v \leftarrow D_0}\left[\varphi(v)|E\right] + r(D_0).
\end{align*}
\end{proof}

\begin{proof}[Proof of Corollary~\ref{cor:alphanonebig}]
Observe first that whenever the variable $X_i(\vec{v})=1$, the item is awarded to bidder $i$, and therefore the payment made is at most $v_i$. But additionally, in order to sell the item to buyer $i$ when $v_i < \beta_i$ requires concealing at least one commitment and paying a fine (otherwise, bidder $i$ expects not to win the item). As the fine charged per concealed commitment is $k_i$, this means that the seller's total revenue is at most $v_i - k_i$. In particular, this also concludes that the seller's total revenue when awarding the item to buyer $i$ when $v_i < k_i$ is non-positive. Therefore, we can write:
$$\mathbb{E}_{\vec{v} \leftarrow D}\left[R(\vec{v}) \cdot I(\forall\ i, v_i < \beta_i)\right] \leq \mathbb{E}_{\vec{v} \leftarrow D}\left[\sum_i (v_i - k_i)\cdot X_i(\vec{v})\right].$$
But now let's consider $\mathbb{E}_{\vec{v} \leftarrow D}\left[ (v_i - k_i) \cdot X_i(\vec{v})\right]$ separately for each $i$. The event $X_i(\vec{v})=1$ satisfies the hypotheses of Lemma~\ref{lem:alphamhr}, as $X_i(\vec{v})=1$ implies that $v_i > k_i \geq r(D_i)$. Therefore, Lemma~\ref{lem:alphamhr} allows us to conclude that:
\begin{align*}
\mathbb{E}_{\vec{v} \leftarrow D}\left[(v_i - k_i)\cdot X_i(\vec{v})\right]
& \leq \mathbb{E}_{\vec{v} \leftarrow D}\left[(\varphi_i(v_i)/\alpha -k_i + r_i) \cdot X_i(\vec{v})\right].
\end{align*}
\end{proof}

\begin{proof}[Proof of Lemma~\ref{lem:alphapricing}]
Starting from Lemma~\ref{lem:alpha-prob-tail}, we get the following chain of inequalities (letting $r:=r(D_0)$ for simplicity of notation):
\begin{align*}
p \cdot \Pr_{v \leftarrow D_0}[v \geq p] &\leq p \cdot \Pr_{v \leftarrow D_0}[v \geq r] \cdot \left(\frac{r}{(1-\alpha)p + \alpha r}\right)^{\frac{1}{1-\alpha}}\\
& = r \cdot  \Pr_{v \leftarrow D_0}[v \geq r] \cdot (p/r) \cdot \left(\frac{r}{(1-\alpha)p + \alpha r}\right)^{\frac{1}{1-\alpha}}\\
&\leq  r \cdot  \Pr_{v \leftarrow D_0}[v \geq r] \cdot (p/r) \cdot \left(\frac{r}{(1-\alpha)p}\right)^{\frac{1}{1-\alpha}}\\
&=r \cdot  \Pr_{v \leftarrow D_0}[v \geq r] \cdot (1-\alpha)^{-\frac{1}{1-\alpha}} \cdot (p/r) \cdot \left(\frac{r}{p}\right)^{\frac{1}{1-\alpha}}\\
&=r \cdot  \Pr_{v \leftarrow D_0}[v \geq r] \cdot (1-\alpha)^{-\frac{1}{1-\alpha}} \cdot\left(\frac{r}{p}\right)^{\frac{\alpha}{1-\alpha}}.\\
\end{align*}
The final line completes the proof.
\end{proof}

\section{Proof of Proposition~\ref{prop:alphaone}}\label{app:alphaone}
\begin{proof}[Proof of Proposition~\ref{prop:alphaone}]

{
When $\alpha = 1$, we can directly invoke Theorem~\ref{thm:mhr}. When $\alpha \in (0,1)$, we will set $f(n,D):= r(D)\bigg(\bigg(\frac{1}{1-\alpha}\bigg)^{\frac{1}{1-\alpha}}\frac{1}{\alpha}\bigg)^{\frac{1-\alpha}{\alpha}}$ for all $n$. We will use $r:=r(D)$ and $k:=f(n,D)$ for ease of notation. Recall that $R(v)$ denotes the revenue of the auctioneer when the single (real) bidder bids $v$, and $\beta$ denotes the effective commitment to the single (real) bidder. Importantly, note that in the single-bidder case, $\beta$ is \emph{fixed} and \emph{not a random variable} (because the auctioneer just decides what fake bids to submit, because there are no other bidders' bids to forward). So, for example, $I(k < \beta)$ is deterministically $1$ or $0$.

\begin{claim}\label{claim:alpha-single-bidder}
$\mathbb{E}_{v \leftarrow D}[R(v) \cdot I(v < \beta)] \leq (\rev(D) - \mathbb{E}[\varphi(v)/\alpha \cdot I(v > \beta)] )\cdot I(k < \beta)$
\end{claim}
\begin{proof}
Define $X(v)$ as the indicator random variable for the event where $k < v < \beta$ and the item is allocated to the real bidder. Observe that the auctioneer charge the buyer at most $v$. Also, if $v < \beta$, then the auctioneer has to conceal at least one bid. In this case, $R(v) \leq v - k$. If $\beta \leq k$, then $v < k$ but if so, the auctioneer gets non-positive revenue. We get:
\begin{align*}
    \mathbb{E}_{v \leftarrow D}[R(v) \cdot I(v < \beta)] &\leq \mathbb{E}_{v \leftarrow D}[(v - k) \cdot X(v)] \cdot I(k < \beta)\\
    &\leq \mathbb{E}_{v \leftarrow D}[\varphi(v)/\alpha \cdot X(v)] \cdot I(k < \beta)
\end{align*}
The first line is simply restating the work above. The second line follows by Lemma~\ref{lem:alphamhr}, observing that $X(v) = 1$ implies $v > k \geq r$. Now observe that if $k \geq \beta$, then $\Pr[X(v)] = 0$ which implies $\mathbb{E}_{v \leftarrow D}[\varphi(v)/\alpha \cdot X(v)] = 0$ and we are done. For the case where $k < \beta$, we can write that $X(v) \leq I(k < v < \beta) = I(v > k) - I(v > \beta)$:
\begin{align*}
    \mathbb{E}_{v \leftarrow D}[\varphi(v)/\alpha \cdot X(v)] &\leq \mathbb{E}_{v \leftarrow D}[\varphi(v)/\alpha \cdot I(v > k)] - \mathbb{E}_{v \leftarrow D}[\varphi(v)/\alpha \cdot I(v > \beta)]\\
    &= \frac{k}{\alpha}Pr_{v \leftarrow D}[v > k] - \mathbb{E}_{v \leftarrow D}[\varphi(v)/\alpha \cdot I(v > \beta)]
\end{align*}
The first line is simply restating the work above and the second line applies Myerson's Lemma for the single item, single bidder auction which sets price $k$. To bound $k\cdot \Pr_{v \leftarrow D}[v > k]$, we directly use Lemma~\ref{lem:alphapricing}:
\begin{align*}
k \cdot \Pr_{v \leftarrow D}[v >k] &\leq  r \cdot \Pr_{v \leftarrow D}[v \geq r] \cdot (1-\alpha)^{-1/(1-\alpha)} \cdot (k/r)^{-\alpha/(1-\alpha)}\\
&\leq r \cdot \Pr_{v \leftarrow D}[v \geq r] \cdot (1-\alpha)^{-1/(1-\alpha)} \cdot \alpha \cdot (1-\alpha)^{\frac{1}{1-\alpha}}\\
&= \alpha \rev(D)
\end{align*}
In the second line, we use the fact $k = r\bigg(\bigg(\frac{1}{1-\alpha}\bigg)^{\frac{1}{1-\alpha}}\frac{1}{\alpha}\bigg)^{\frac{1-\alpha}{\alpha}}$. Putting both bounds together, we conclude the proof of the statement.
\end{proof}
Combining Claim~\ref{claim:alpha-single-bidder} and Lemma~\ref{lem:onebig}, we get:
\begin{align*}
    \mathbb{E}_{v \leftarrow D}[R(v)] &= \mathbb{E}_{v \leftarrow D}[R(v) \cdot I(v > \beta)] + \mathbb{E}_{v \leftarrow D}[R(v) \cdot I(v < \beta)]\\
    &\leq \mathbb{E}[\varphi(v) \cdot I(v > \beta)] + (\rev(D) - \mathbb{E}[\varphi(v)/\alpha \cdot I(v > \beta)])\cdot I(k < \beta)\\
    &\leq \rev(D)
\end{align*}
The first line follows from linearity of expectation. The second line directly applies Lemma~\ref{lem:onebig} and Claim~\ref{claim:alpha-single-bidder}. In the last line, we observe that if $k \geq \beta$, then we can use the fact $\mathbb{E}[\varphi(v) \cdot I(v > \beta)] \leq \rev(D)$, else $\mathbb{E}[\varphi(v) \cdot I(v > \beta)] - \mathbb{E}[\varphi(v)/\alpha \cdot I(v > \beta)] < 0$ because $\alpha \in (0, 1)$.
}
\end{proof}

\section{Example: DRA on Heavy Tail Distributions}\label{app:regular}
\begin{proposition}\label{prop:ultra}
There exists a regular distribution $D$ with unbounded expected value, such that for some $f(n, D)$, $\dra(f)$ is optimal, strategyproof and computationally $2/3$-credible when there is a single (real) buyer from $D$.
\end{proposition}
\begin{proof}
We will define $D$ such that even though it has infinite expected value, the tail of $D$ is not too heavy so that with a constant fine we can limit the revenue the auctioneer can obtain. We will use $f(n, D) = e^e$. For ease of notation $k := f(n, D)$ and $r := r(D)$ is the optimal reserve. Assume there is a single real buyer from $D$.

Let's first define an extension of tetration for positive real numbers.
$$h(x) := \begin{cases}1 + x &\text{, for $-1 < x \leq 0$}\\
e^{h(x-1)} & \text{, for $x > 0$}\end{cases}$$
This function is continuous and differentiable in $(-1, \infty)$ (\cite{hooshmand2006ultra} for a detailed analysis of ultra exponential functions). Differentiating with respect to $x$,
$$h'(x) = \begin{cases}1 & \text{, for $-1 < x \leq 0$}\\
h(x)h'(x-1) & \text{, for $x > 0$}\end{cases}$$
Define the natural super-logarithm $\slog(\cdot)$ to be the inverse of $h(\cdot)$. More formally, $\slog(x) = y$ if and only if $h(y) = x$ which implies the following property for every $x \in (0, \infty)$,
$$\slog(y) = \begin{cases}
y-1 \quad & \text{for $0 < y \leq 1$}\\
1 + \ln^*(\ln(y)) \quad & \text{for $y > 1$}
\end{cases}$$
Observe that $\slog(e^x) = 1 + \slog(x)$ and $\slog(1) = 0$. Informally, one can interpret $\slog(x)$ as counting how many times one must take the natural-logarithm of $x$ to get 1. We define the distribution $D$ supported on $[1, \infty)$ in terms of the CDF:
$$Pr_{v \leftarrow D}[v > x] := \begin{cases} 1 & \text{, for $x \leq 1$}\\ \frac{d}{dx}\slog(x) & \text{, for $x > 1$}\end{cases}$$
By the chain rule, for $x > 1$, $Pr_{v \leftarrow D}[v > x] = \frac{1}{h'(\slog(x))}$. To see this is a valid distribution, observe $Pr_{v \leftarrow D}[v > x]$ is monotone decreasing, $Pr_{v \leftarrow D}[v > 1] = 1$, and $\lim_{n \to \infty}Pr_{v \leftarrow D}[v > n] = 0$. Now observe $D$ has unbounded expected value:
\begin{align*}
E_{v \leftarrow D}[v] &= 1 + \int_{1}^\infty Pr_{v \leftarrow D}[v > x] dx\\
    &= 1 + \lim_{n \to \infty} \slog n - \slog 1 = \infty
\end{align*}
In the next claim, we show that for every price $p \in \mathbb R^+$, $p Pr_{v \leftarrow D}[v > p] \leq 1$.
\begin{claim}\label{claim:counterexample-1-marginal-revenue}
For all $p \in \mathbb R^+$, $p Pr_{v \leftarrow D}[v \geq p] \leq 1$.
\end{claim}
\begin{proof}
If $p \leq 1$, then $p  Pr_{v \leftarrow D}[v \geq p] = p \leq 1$. For $p > 1$, $\slog p > 0$ and by the recursive definition of $h'(\cdot)$, we can expand $h'(\slog p)$:
\begin{align*}
   Pr_{v \leftarrow D}[v > p] &= \frac{1}{h'(\slog p)} = \frac{1}{h(\slog p)h'(\slog p - 1)}\\
    &= \frac{1}{p h'(\slog p - 1)}\\
    &\leq \frac{1}{p}
\end{align*}
where the last inequality follows from the fact $(\ln^*p - 1) > -1$ and $h'(\ln^*p - 1) \geq 1$.
\end{proof}

The claim above implies that $r(D) = 1$ is the optimal reserve since $Pr_{v \leftarrow D}[v > r(D)] = 1$. Let $x_1 < x_2 < ... < x_n$, be the bids sent to the real bidder. For easy of notation, define $x_0 = 0$ and $x_{n+1} = \infty$.

\begin{observation}\label{obs:log-example}
If $n \geq 1$, then $E_{v \leftarrow D}[R(v) \cdot I(0 \leq v < x_1)] = 0$.
\end{observation}
\begin{proof}
If $v < x_1$, the auctioneer has to hide at least $x_1$ which implies the revenue is non-positive since $k > r(D)$.
\end{proof}

\begin{claim}
$E_{v \leftarrow D}[R(v)] \leq \max\{1, n\}$
\end{claim}
\begin{proof}
If $n = 0$, the auctioneer is simply implementing $\dra(f)$ in earnest and the revenue is 1. Assuming $n \geq 1$, observe that if $v \in [x_i, x_{i+1})$ the auctioneer obtains revenue at most $x_i$ by concealing all $x_j > x_i$ and revealing $x_j \leq x_i$. We get:
\begin{align*}
    E_{v \leftarrow D}[R(v)] &\leq \sum_{i = 1}^n E_{v \leftarrow D}[R(v) \cdot I(x_i \leq v < x_{i+1})]\\
    &\leq \sum_{i = 1}^n E_{v \leftarrow D}[x_i \cdot I(v \geq x_i)]\\
    &\leq n
\end{align*}
The first line is linearity of expectation and Observation~\ref{obs:log-example}. The second line uses the fact $I(x_i \leq v < x_{i+1}) \leq I(v \geq x_i)$ and restating the work above. The last line is Claim~\ref{claim:counterexample-1-marginal-revenue}.
\end{proof}

The auctioneer gets revenue $1$ by implementing the auction in earnest, and the claim above states that the auctioneer needs to submit at least $3$ fake bids to get revenue bigger or equal than $3$. Next, We argue that the auctioneer cannot get revenue bigger than $3$, no matter the number of fake bids sent to the real bidder.

So assuming $n > 3$, we next argue that it is without loss of generality to assume that $x_1 > n$. Consider the event where $v \in [x_1, x_2)$ and recall that to charge the real bidder $x_1$, the auctioneer has to conceal $n-1$ bids and pay $k\cdot(n-1) \geq n$ for our choice of $k$. If $x_1 \leq n$, then the auctioneer gets non-positive revenue. We conclude that the auctioneer cannot decrease their revenue by not bidding on $x_1$. So from now on, we will assume that the smallest bid $x_1 > n$.

\begin{claim}\label{claim:log-bound-1}
$E_{v \leftarrow D}[R(v) \cdot I(0 \leq v < e^n)] \leq 2$
\end{claim}
\begin{proof}
Observe that $R(v) \leq v$. We get:
\begin{align*}
    E_{v \leftarrow D}&[R(v) \cdot I(0 \leq v < e^n)] = E_{v \leftarrow D}[R(v) \cdot I(0 \leq v < x_1)] + E_{v \leftarrow D}[R(v) \cdot I(x_1 \leq v < e^n)]\\
    &\leq E_{v \leftarrow D}[v \cdot I(n \leq v < e^n)]\\
    &=\int_0^\infty Pr_{v \leftarrow D}[v \cdot I(n \leq v < e^n) > x] dx\\
    &= \int_{0}^n Pr_{v \leftarrow D}[v \cdot I(n \leq v < e^n) > x] dx + \int_{n}^{e^n} Pr_{v \leftarrow D}[v \cdot I(n \leq v < e^n) > x] dx\\
    &\quad + \int_{e^n}^\infty Pr_{v \leftarrow D}[v \cdot I(n \leq v < e^n)> x] dx\\
    &\leq \int_0^n Pr_{v \leftarrow D}[n \leq v < e^n]dx + \int_{n}^{e^n}Pr_{v \leftarrow D}[v > x]dx + \int_{e^{n}}^\infty Pr_{v \leftarrow D}[n \leq v < e^n | v > x] dx\\
    &\leq n Pr_{v \leftarrow D}[v \geq n] + (\ln^*(e^n) - \ln^*(n)) + 0\\
    &\leq 1 + (1 + \ln^* n - \ln^* n) = 2
\end{align*}
The first line is linearity of expectation. The second line is Observation~\ref{obs:log-example}, restating that $R(v) \leq v$ and $x_1 > n$. The third line is just the integral form of expectation for non-negative random variables. The fourth line is simply linearity of integration. In the fifth line, we observe $Pr_{v \leftarrow D}[v \cdot I(n \leq v < e^n) > x] = Pr_{v \leftarrow D}[v > x, n \leq v < e^n]$. In the sixth line, the first term follows by integrating the constant function in the interval $0$ to $n$, the second term follows from the fundamental theorem of calculus, and the third term is 0 simply by observing the event $\{n \leq v < e^n | v > x > e^n\} = \emptyset$. The last line follows by Claim~\ref{claim:counterexample-1-marginal-revenue} and the recursive definition of super-logarithm.
\end{proof}

\begin{claim}\label{claim:log-bound-2}
$E_{v \leftarrow D}[R(v) \cdot I(v \geq e^n)] \leq 1$
\end{claim}
\begin{proof}
Let $m = \min \{i : x_i \geq e^n\}$. For all $i \geq m$, when $v \in [x_i, x_{i+1})$, $R(v) \leq x_i$ which happens with probability at most $Pr_{v \leftarrow D}[v \geq x_i]$.
\begin{align*}
    E_{v \leftarrow D}[R(v) \cdot I(v \geq e^n)]&\leq \sum_{i = m}^n x_i Pr_{v \leftarrow D}[v \geq x_i]\\
    &= \sum_{i = m}^n \frac{x_i}{h'(\ln^* x_i)}\\
    &=\sum_{i = m}^n \frac{x_i}{x_i \ln x_i h'(\slog x_i - 2)}\\
    &\leq \sum_{i = m}^n \frac{1}{n h'(\slog x_i - 2)}\\
    &\leq \frac{n-m+1}{n}\\
    &\leq 1
\end{align*}
The first line is restating the work above. The second line uses the definition of $D$. The third line uses the fact $x_i \geq e^n \geq e^e$, which implies $h'(\ln^* x_i) = h(\ln^* x_i)h(\ln^* x_i - 1)h'(\ln^* x_i - 2) = x_i \ln x_i h'(\ln^* x_i - 2)$. The fourth line uses $x_i \geq e^n$. The fifth lines uses $x_i \geq e^e$, which implies $h'(\ln^* x_i - 2) \geq 1$.
\end{proof}

Combining the bounds in Claim~\ref{claim:log-bound-1} and Claim~\ref{claim:log-bound-2}, the sequence of inequalities witnesses that the auctioneer can obtain revenue at most $3$ when $f(n, D) = e^e$. Because Myerson's optimal auction gets revenue $1$, we conclude $\dra(f)$ is $2/3$-credible when there is a single bidder from $D$.
\end{proof}

\section{Example: DRA on $\alpha$-Strongly Regular Distributions}\label{app:malleable}
\begin{definition}[Almost reasonable] A commitment is \emph{loosely tied} to $(m_1,r_1),\ldots, (m_k, r_k)$ if the participant who sent $c$ explicitly computed a poly-time function of $c_1,\ldots, c_k$, which were explicitly tied to $(m_1,r_1),\ldots, (m_k,r_k)$. A deviation is \emph{almost reasonable with respect to $g$} for the auctioneer in the communication game if whenever the auctioneer reveals a commitment to $c$, with $c= \commit(m, r)$, then $m=g(m_1,\ldots, m_k)$ for some $(m_1,r_1),\ldots, (m_k,r_k)$ to which $c$ is loosely tied.
\end{definition}

To help get intuition for the definition, observe that any reasonable deviation is also almost reasonable with respect to the identity function. If the commitment scheme is malleable, though, and it is possible to compute a poly-time function $g(\cdot,\ldots, \cdot)$ on un-revealed commitments, then a deviation which forwards a commitment to $g(b_1,\ldots, b_{n-1})$ to bidder $n$ (before the commitments are revealed), and then later reveals $g(b_1,\ldots, b_{n-1})$ is almost reasonable with respect to $g$. 

\begin{definition}[Strongly Computationally Credible] An auction is \emph{strongly computationally credible with respect to $g$} if, in expectation over $\vec{v} \leftarrow D$, and buyers being truthful, the auctioneer maximizes their expected revenue, over all deviations which are both safe and almost reasonable with respect to $g$, by executing the auction in earnest.

An auction is strongly $\varepsilon$-computationally credible with respect to $g$ if executing the auction in earnest yields a $(1-\varepsilon)$-fraction of the expected revenue of any safe, almost reasonable with respect to $g$, deviation.
\end{definition}

Theorems~\ref{thm:malleable} and~\ref{thm:alphabounded} below show that, perhaps surprisingly, both assumptions of unbounded distributions and non-malleable commitment schemes are necessary for Theorem~\ref{thm:alphamhr} (whereas they were not necessary for Theorem~\ref{thm:mhr}, and appear to be just a technical artifact of our proof approach through Lemma~\ref{lem:all}). 

\begin{theorem}\label{thm:malleable} Let $g$ be any function such that $g(m_1,\ldots, m_k) \geq \max_{i \in [k]}\{m_i\}$. Then for all $\alpha < 1$, there exists a $D_0$ which is $\alpha$-strongly regular, such that for all $\varepsilon > 0$ and all $f(\cdot, \cdot)$, there exists a sufficiently large $n$ such that $\dra(f)$ is not strongly $(1-\alpha-\varepsilon)$-computationally credible with respect to $g$ for the instance $D:=\times_{i=1}^n D_0$.
\end{theorem}

In Theorem~\ref{thm:alphabounded} below, by the notation $D_0^T$, we mean the distribution $D_0$ \emph{truncated at $T$} (all probability mass above $T$ is moved to $T$. Formally, $\Pr_{v \leftarrow D_0}[v \geq x] = \Pr_{v \leftarrow D_0^T}[v \geq x]$ for all $x \leq T$, and $\Pr_{v \leftarrow D_0^T}[v \geq x] = 0$ for all $x > T$). 

\begin{theorem}\label{thm:alphabounded} For all $\alpha < 1$, there exists a distribution $D_0$ which is $\alpha$-strongly regular, such that for all $\varepsilon > 0$ and all $f(\cdot, \cdot)$ such that $f(n,D)$ depends on $\alpha, \varepsilon, r(D), n$ (but may depend arbitrarily on these values), there exists a sufficiently large $T$ and $n$ such that $\dra(f)$ is not $(1-\alpha-\varepsilon)$-credible for the instance $D:= \times_{i=1}^n D^T_0$.
\end{theorem}

The proofs of Theorems~\ref{thm:malleable} and~\ref{thm:alphabounded} will be nearly identical, and just wrap up differently at the end. Specifically, the $D_0$ for both theorems is the same, and the deviations are the same in spirit. First, we repeat the $D_0$ we use, which is the same from the proof of Theorem~\ref{thm:alphamhrneg}, repeated below for completeness:

\begin{minipage}{.5\linewidth}
$$F^\alpha(v) = \begin{cases}0 \quad & \text{, $v < 1$}\\
1 - \big(\frac 1 v\big)^{\frac 1 {1-\alpha}} \quad & \text{, $v \geq 1$} \end{cases}$$
\end{minipage}
\begin{minipage}{.5\linewidth}
$$f^\alpha(v) = \begin{cases} 0 \quad &\text{, $v < 1$}\\
\frac{1}{1-\alpha}\big(\frac 1 v\big)^{\frac {2-\alpha}{1-\alpha}} \quad &\text{, $v \geq 1$}\end{cases}$$
\end{minipage}

The hazard rate of $F^\alpha$ is $h^{F^\alpha}(v) = \frac{1}{(1-\alpha) v}$ for $v \geq 1$ and the virtual value function of $F^\alpha$ is $\varphi^{F^\alpha}(v) = v - \frac{1}{h^{F^\alpha}(v)} = \alpha v$, so $D_0$ is $\alpha$-strongly regular. 

Our proof will proceed as follows. First, we will quickly establish that $\rev(D) = \alpha \cdot \mathbb{E}_{\vec{v} \leftarrow D}[\max_i \{v_i\}]$. Then, we will design a deviation which achieves revenue arbitrarily close to $\mathbb{E}_{\vec{v} \leftarrow D}[\max_i \{v_i\}]$ for sufficiently large $n$. Finally, we will show that this deviation satisfies the properties of the two theorems (separately). 

\begin{claim} For any $n$, and $D:= D_0^n$, $\rev(D) = \alpha \cdot \mathbb{E}_{\vec{v} \leftarrow D}[\max_i \{v_i\}]$. 
\end{claim}
\begin{proof}
This could be verified by direct (but tedious) calculations. A simpler proof observes that $\rev(D)= \mathbb{E}[\max_i \{\varphi(v_i)\}]$ by Myerson's lemma. But because $\varphi(v_i):= \alpha v_i$, we immediately get that $\rev(D)= \mathbb{E}[\max_i \{\varphi(v_i)\}] = \alpha \cdot \mathbb{E}[\max_i \{v_i\}]$. 
\end{proof}

Now, consider the following deviation in $\dra(f)$. For any desired parameter $\delta$, we will argue that the following deviation gets revenue (not counting fines) at least $(1-5\delta)\cdot \mathbb{E}[\max_i \{v_i\}]$, and also pays at most $\poly(1/\delta)$ fines. Importantly, \emph{the number of fines paid will depend only on $\delta$ and is independent of $n$} (it is not crucial that it is polynomial in $1/\delta$). 

Our deviation is as follows: for {$\ell=0$} to $z$, where $z:= \frac{(1-\alpha) \cdot \ln_{1+\delta}(1/(\alpha\delta  {\cdot \delta^{\alpha/(1-\alpha)}}))}{{\alpha}}$, let $x_\ell:= \delta \cdot (1+\delta)^\ell \cdot n^{1-\alpha}$. Upon receiving commitments to $b_1,\ldots, b_n$ from the bidders, the auctioneer will send to bidder $i$:
\begin{itemize}
\item A commitment to $x_\ell$, for all $\ell$.
\item A single additional commitment to $y^*_i$, where it is guaranteed that $b_i > y^*_i$ \emph{only if} $b_i > b_j$ for all $j \neq i$ (but perhaps $b_i \leq y^*_i$ for all $i$, this is also fine). Note that we defer to the last step of each proof exactly how to accomplish this for the two theorems of interest. 
\end{itemize}

After asking all bidders to reveal, the auctioneer will reveal its own commitments as follows:
\begin{itemize}
\item To all bidders $j \neq \arg\max_i \{b_i\}$, reveal all commitments. Note that this guarantees that bidder $j$ does not expect to win the item.
\item To bidder $i^* := \arg\max_i \{b_i\}$: 
\begin{itemize}
\item If $b_{i^*} < x_0$, reveal all commitments. 
\item If $b_{i^*} > x_{z}$, reveal all commitments except $y^*_{i^*}$.
\item If $b_{i^*} \in [x_\ell,x_{\ell+1})$, reveal commitments $x_0,\ldots, x_\ell$ and conceal the rest.
\end{itemize}
\end{itemize}

Now, we analyze the revenue of this deviation \emph{excluding fines}. We first need a quick observation about the relationship between $n^{1-\alpha}$ and $\mathbb{E}[\max_i \{v_i\}]$.

\begin{observation}\label{obs:quick} $\mathbb{E}[\max_i \{v_i\}] \geq (1-1/e)n^{1-\alpha}$. 
\end{observation}
\begin{proof}
Observe that $\Pr[\max_i \{v_i\} > n^{1-\alpha}] \geq (1-1/e)$ (because each $v_i$ exceeds $n^{1-\alpha}$ independently with probability exactly $1/n$). Therefore, $\mathbb{E}[\max_i \{v_i\}] \geq (1-1/e) n^{1-\alpha}$.
\end{proof}

\begin{lemma}\label{lem:deviation} Excluding fines, the deviation above guarantees revenue at least $(1-5\delta)\cdot \mathbb{E}[\max_i \{v_i\}]$ in expectation.
\end{lemma}
\begin{proof}
Observe first that the payment made by bidder $i^*$ is:
\begin{itemize}
\item At least $0$, when $v_{i^*} \notin [x_0,x_{z}]$.
\item At least $(1-\delta)\cdot v_{i^*} \in [x_0,x_{z}]$. 
\end{itemize}
This is just because payments (excluding fines) are always non-negative, and because whenever $v_{i^*} \in [x_0,x_{z}]$, there is always a revealed commitment within $(1-\delta)\cdot v_{i^*}$. So our only task is to argue that $\mathbb{E}[\max_i \{v_i\} \cdot I(\max_i \{v_i\} \in [x_0,x_{z}])] \geq (1-4\delta) \cdot \mathbb{E}[\max_i \{v_i\}]$. 

To do this, we will simply argue that the welfare lost from cases where $v_{i^*} < x_0$ is at most $2\delta \cdot \mathbb{E}[\max_i \{v_i\}]$, and also that the welfare lost from cases where $v_{i^*} > x_{z}$ is at most $2\delta \cdot \mathbb{E}[\max_i \{v_i\}]$. 

\begin{observation}$\mathbb{E}[\max_i \{v_i\}\cdot I(\max_i\{v_i\} < x_0)]\leq 2\delta \cdot \mathbb{E}[\max_i \{v_i\}]$.
\end{observation}
\begin{proof}
Simply observe that $\mathbb{E}[\max_i \{v_i\} \cdot I(\max_i\{v_i\} < x_0)] \leq x_0 = \delta n^{1-\alpha} \leq \frac{e}{e-1} \cdot \delta \cdot \mathbb{E}[\max_i \{v_i\}]$.
\end{proof}

\begin{lemma} $\mathbb{E}[\max_i \{v_i\} \cdot I(\max_i \{v_i\} > x_{z})] \leq 2\delta \cdot \mathbb{E}[\max_i \{v_i\}]$.
\end{lemma}
\begin{proof}
The lemma will follow from a few observations. First, we observe that $\mathbb{E}[\max_i \{v_i\} \cdot I(\max_i \{v_i\} > x_{z})] \leq \mathbb{E}[\sum_i v_i \cdot I(v_i > x_{z})] = n \cdot \mathbb{E}[v_1 \cdot I(v_1 > x_{z})]$. Next, we observe that $n \cdot \mathbb{E}[v_1 \cdot I(v_1 > x_{z})] = n \cdot x_{z} \cdot (1-F^\alpha(x_{z}))/\alpha$. This could be observed by direct (but tedious) calculations. Alternatively, it follows as the expected revenue of setting price $x_{z}$ to bidder one is exactly $x_{z} \cdot (1-F^\alpha(x_{z}))$, but also $\alpha \cdot \mathbb{E}[v_1 \cdot I(v_1 > x_{z})]$ by Myerson's lemma (and that $\varphi(v)= \alpha v$). 

Now, our job is just to upper bound $x_{z} \cdot (1-F^\alpha(x_{z}))$. Recalling the definition of $F^\alpha(\cdot)$, this is exactly $x_{z}^{1-\frac{1}{1-\alpha}} = x_{z}^{-\frac{\alpha}{1-\alpha}}$. But now we get:
\begin{align*}
x_{z}^{-\frac{\alpha}{1-\alpha}} &= ({\delta} \cdot n^{1-\alpha}\cdot (1+\delta)^{z})^{-\frac{\alpha}{1-\alpha}}\\
&= {\delta^{-\alpha/(1-\alpha)}} \cdot n^{-\alpha} \cdot (1+\delta)^{-\frac{\frac{(1-\alpha) \cdot \ln_{1+\delta}(1/(\alpha\delta \cdot {\delta^{\alpha/(1-\alpha)}}))}{\alpha} \cdot \alpha}{1-\alpha}}\\
&= {\delta^{-\alpha/(1-\alpha)}} \cdot n^{-\alpha} \cdot (1+\delta)^{-\ln_{1+\delta}(1/(\alpha \delta \cdot {\delta^{\alpha/(1-\alpha)}}))}\\
&= {\delta^{-\alpha/(1-\alpha)}} \cdot n^{-\alpha}\cdot \alpha \cdot \delta \cdot {\delta^{\alpha/(1-\alpha)}}.
\end{align*}

Finally, we now conclude that:
\begin{align*}
n \cdot x_{z} \cdot (1-F^\alpha(x_{z}))/\alpha &= n \cdot n^{-\alpha} \cdot (\alpha\delta) /\alpha\\
&= \delta n^{1-\alpha}.
\end{align*}

Again by Observation~\ref{obs:quick}, this is now at most $2\delta \mathbb{E}[\max_i \{v_i\}]$. 
\end{proof}

Now, this concludes the proof of Lemma~\ref{lem:deviation}. We have just argued that $\mathbb{E}[\max_i \{v_i\}\cdot I(\max_i\{v_i\} < x_0)]\leq 2\delta \cdot \mathbb{E}[\max_i \{v_i\}]$, and also that $\mathbb{E}[\max_i \{v_i\} \cdot I(\max_i \{v_i\} > x_{z})] \leq 2\delta \cdot \mathbb{E}[\max_i \{v_i\}]$. Therefore, $\mathbb{E}[\max_i \{v_i\} \cdot I(\max_i \{v_i\} \in [x_0,x_{z}])] \geq (1-4\delta) \cdot \mathbb{E}[\max_i \{v_i\}]$, and the lemma follows as $(1-\delta)\cdot (1-4\delta) \geq (1-5\delta)$.
\end{proof}

Lemma~\ref{lem:deviation} argues that the revenue \emph{excluding fines} of the above deviation is large. But we must now argue that the fines paid are small.

\begin{lemma}\label{lem:fines} The total fines paid by the above deviation is at most $z\cdot f(z+2,D_0)$. Importantly, observe that this is \emph{independent of $n$}.
\end{lemma}
\begin{proof}
The total number of commitments sent to each player is $z+1$. So including themselves, each player believes there are only $z+2$ bidders, and fines are computed accordingly. To each player, the auctioneer then clearly pays at most $z$ fines when they conceal commitments (because they send only $z+1$, and reveal at least one when they hide any). 
\end{proof}

We can now wrap up the proofs of Theorems~\ref{thm:malleable} and~\ref{thm:alphabounded}.

\begin{proof}[Proof of Theorem~\ref{thm:malleable}]
To implement the proposed deviation, simply set $y_i^* = g(\vec{b}_{-i})$. By definition of $g$, this guarantees that $y_i^* > b_i$ for all $i \neq i^*$. For a given $\varepsilon$, first set $\delta < \varepsilon/100$ and use the proposed deviation. This guarantees revenue (including fines) of at least $(1-\varepsilon/2)\cdot \mathbb{E}[\max_i \{v_i\}] - z\cdot f(z+2,D_0)$. Because $z\cdot f(z+2,D_0)$ \emph{is independent of $n$}, and $D_0$ is unbounded, there exists a sufficiently large $n$ such that $z\cdot f(z+2,D_0) \leq (\varepsilon/2) \cdot \mathbb{E}[\max_i \{v_i\}]$. This completes the proof.
\end{proof}

\begin{proof}[Proof of Theorem~\ref{thm:alphabounded}]
Consider $D_0^T$ for $T$ to be set later, and define $D^T:=\times_{i=1}^n D_0^T$. To implement the proposed deviation, simply set $y_i^* = T+1$. For a given $\varepsilon$, first set $\delta < \varepsilon/100$ and then use the proposed deviation. This guarantees revenue (including fines) of at least $(1-\varepsilon/2)\cdot \mathbb{E}_{\vec{v} \leftarrow D_0^T}[\max_i \{v_i\}] - z\cdot f(z+2,D^T_0)= (1-\varepsilon/2)\cdot \mathbb{E}_{\vec{v} \leftarrow D_0^T}[\max_i \{v_i\}] - z\cdot f(z+2,D_0)$. Note that the last equality follows as we have assumed that $f(a,D)$ depends only on $a,\varepsilon,\alpha,r(D)$, and $r(D)=r(D_0)=1$. But now because $z\cdot f(z+2,D_0)$ \emph{is independent of $n$}, and $D_0$ is unbounded, there exists a sufficiently large $n$ and $T$ such that $z\cdot f(z+2,D_0) \leq (\varepsilon/2) \cdot \mathbb{E}_{\vec{v} \leftarrow D_0^T}[\max_i \{v_i\}]$. This completes the proof.
\end{proof}


Finally, we show that Theorems~\ref{thm:malleable} and~\ref{thm:alphabounded} are tight. In the proof of Theorem~\ref{thm:malleable2} below, note that Lemma~\ref{lem:onebig} and Corollary~\ref{cor:alphanonebig} only require that $\beta_i$ is a function of $\vec{b}_{-i}$, and not that it takes a particular (non-malleable) form.

\begin{theorem}\label{thm:malleable2}
Let $f(n, D_i) := r(D_i)$. Then when all $D_i$ are $\alpha$-strongly regular (bounded or unbounded), $DRA(f)$ is optimal, strategyproof, computationally $(1-\alpha)$-credible, and strongly computationally $(1-\alpha)$-credible with respect to all functions $g$.
\end{theorem}
\begin{proof}
Let $k_i = r(D_i)$. Recall $X_i(\vec v)$ is the indicator variable for the event where the item is allocated to buyer $i$ and $k_i \leq v_i < \beta_i$. Also observe that $\bar \phi_i(v_i) = \phi_i(v_i)$ since $D_i$ is regular. If $X_i(\vec v) = 1$, then $v_i \geq k_i \geq r_i$ and by definition $\phi_i(v_i) \cdot X_i(\vec v) \geq 0$. We first combine Lemma~\ref{lem:onebig} and Corollary~\ref{cor:alphanonebig}:
\begin{align*}
    \mathbb{E}_{\vec{v} \leftarrow D}[R(\vec{v})] &=\mathbb{E}_{\vec{v} \leftarrow D}[R(\vec{v}) \cdot I(\exists\ i, v_i > \beta_i)] + \mathbb{E}_{\vec{v} \leftarrow D}[R(\vec{v}) \cdot I(\forall\ i, v_i < \beta_i)]\\
    &\leq\mathbb{E}_{\vec{v} \leftarrow D}[\max_i\{\varphi_i(v_i)\} \cdot I(\exists\ i, v_i > \beta_i)] + \sum_{i=1}^n\mathbb{E}_{\vec{v} \leftarrow D}[(\varphi_i(v_i)/\alpha + r(D_i) - k_i) \cdot X_i(\vec v)]\\
    &\leq \mathbb{E}_{\vec{v} \leftarrow D}[\max_i\{\varphi_i(v_i)\}] + \sum_{i = 1}^n \mathbb{E}_{\vec{v} \leftarrow D}[\varphi_i(v_i)/\alpha\cdot X_i(\vec v)]\\
    &\quad - \mathbb{E}_{\vec{v} \leftarrow D}[\max_i\{\varphi_i(v_i)\} \cdot I(\forall i, v_i < \beta_i)]\\
    &\leq \rev(D) + \sum_{i = 1}^n\mathbb{E}_{\vec{v} \leftarrow D}[\varphi_i(v_i)/\alpha \cdot X_i(\vec v)] - \sum_{i=1}^n \mathbb{E}_{\vec{v} \leftarrow D}[\varphi_i(v_i) \cdot X_i(\vec v) \cdot I(\forall i, v_i < \beta_i)]\\
    &\leq \rev(D) + \sum_{i = 1}^n\mathbb{E}_{\vec{v} \leftarrow D}[(1-\alpha)\varphi_i(v_i)/\alpha \cdot X_i(\vec v)]\\
    &\leq \rev(D) + \frac{1-\alpha}{\alpha} \rev(D)\\
    &= \frac{1}{\alpha} \rev(D)
\end{align*}

The first line is just linearity of expectation. The second line combines Lemma~\ref{lem:onebig} and Corollary~\ref{cor:alphanonebig}. The third line is again by linearity of expectation and the fact $k_i = r(D_i)$. The fourth line uses the fact that $\mathbb{E}_{\vec{v} \leftarrow D}[\max_i\{\varphi_i(v_i)\}]$ is at most the revenue of Myerson's optimal auction and $\max(\varphi_i(v_i)) \geq \sum_{i = 1}^n \varphi_i(v_i)\cdot X_i(\vec v)$ since the indicator random variable $X_i(\vec v) = 1$ for at most one bidder. The fifth line is just the fact $X_i(\vec v) = 1$ implies $\forall i, v_i < \beta_i$ and by linearity of expectation. In the sixth line, we observe that the second term is at most ($\frac{1-\alpha}{\alpha}$ times) the revenue of Myerson's optimal auction.

This concludes that the revenue of any deviation which is safe, and almost reasonable \emph{with respect to any $g(\cdot)$}, is at most $\rev(D)/\alpha$.

\end{proof}

\end{document}